%% file: arxiv.tex
\documentclass[a4paper,UKenglish,cleveref,autoref,thm-restate,numberwithinsect]{lipics-v2021}
\pdfoutput=1
\hideLIPIcs
\input{packages}
\usepackage{arxiv-lipics-v2021} 
\title{A Note on the Parameterised Complexity \mbox{of Coverability in Vector Addition Systems}}
\titlerunning{On the Parameterised Complexity of Coverability in VAS}

\author{Micha\l{} {Pilipczuk}}{Institute of Informatics, University of
  Warsaw, Poland \and \url{https://www.mimuw.edu.pl/~mp248287/}}{michal.pilipczuk@mimuw.edu.pl}{https://orcid.org/0000-0001-7891-1988}{\flag{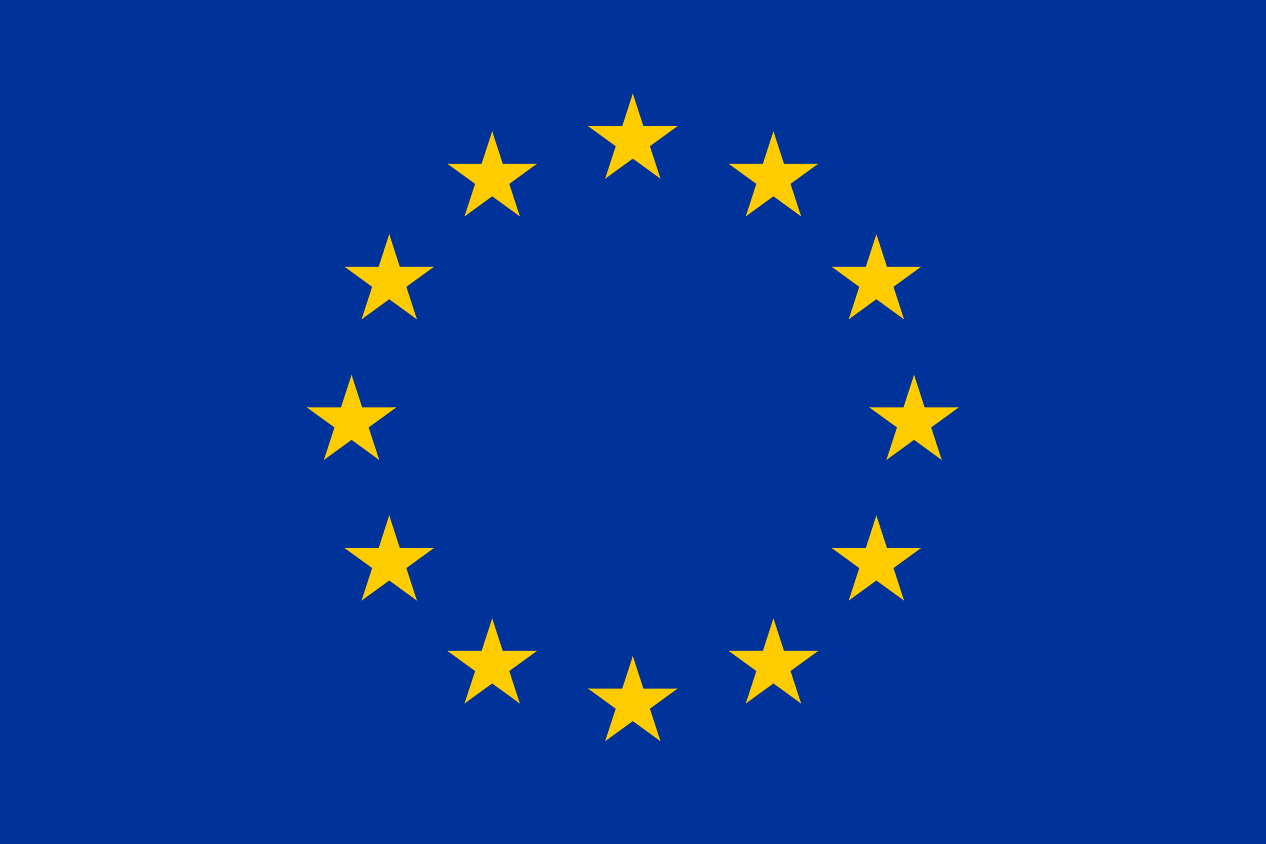}\flag{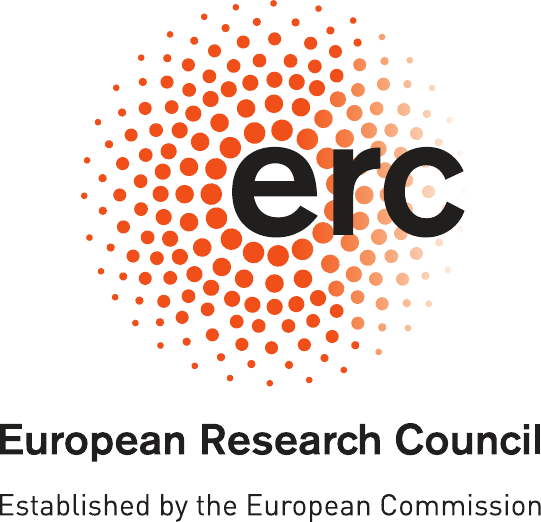}This
  work is a part of project BOBR that has received funding from the
  European Research Council (ERC) under the European Union’s Horizon
  2020 research and innovation programme (grant agreement No. 948057).}

\author{Sylvain {Schmitz}}{Universit\'e Paris Cit\'e, CNRS, IRIF, France \and \url{https://www.irif.fr/~schmitz} }{schmitz@irif.fr}{https://orcid.org/0000-0002-4101-4308}{}

\author{Henry {Sinclair-Banks}}{Institute of Informatics, University of Warsaw, Poland \and \url{http://henry.sinclair-banks.com}}{hsb@mimuw.edu.pl}{https://orcid.org/0000-0003-1653-4069}{This
  work is a part of project INFSYS that has received funding from the
  European Research Council (ERC) under the European Union’s Horizon
  2020 research and innovation programme (grant agreement No. 950398).}%

\authorrunning{Mi. Pilipczuk, S. Schmitz, and H. Sinclair-Banks} %

\Copyright{Micha\l{} Pilipczuk, Sylvain Schmitz, and Henry
  Sinclair-Banks} %



\relatedversiondetails[linktext={10.4230/LIPIcs.IPEC.2025.24}]{Conference Version}{https://doi.org/10.4230/LIPIcs.IPEC.2025.24}

\acknowledgements{This research was initiated during the InfAut~2024
  workshop in Warlity Wielkie, Poland.  The authors thank the
  organisers for launching such a fruitful event.  They also thank an
  anonymous reviewer from IPEC'25 for pointing them to a relevant discussion in Hack's PhD thesis~\cite{Hack76}.}%

\nolinenumbers %

\EventEditors{John Q. Open and Joan R. Access}
\EventNoEds{2}
\EventLongTitle{42nd Conference on Very Important Topics (CVIT 2016)}
\EventShortTitle{CVIT 2016}
\EventAcronym{CVIT}
\EventYear{2016}
\EventDate{December 24--27, 2016}
\EventLocation{Little Whinging, United Kingdom}
\EventLogo{arxiv-logo.pdf}
\SeriesVolume{42}
\ArticleNo{23}

\begin{document}

\input{macros}

\maketitle

\begin{abstract}
    \input{abstract}

\end{abstract}

\section{Introduction}
\input{sections/intro}

\section{Preliminaries}
\input{sections/preliminaries}
\subparagraph{Petri nets.}
\label{app:petri-nets}
\input{appendix/petri-nets}
\input{sections/known-complexity}
\subsection{Further Related Work}
\label{sec:related}
\input{appendix/related}

\section{Unary Encoding}
\label{sec:p-dim}
\input{sections/p-dim}
\subsection{Proof of
  \texorpdfstring{\cref{cl:correctness}}{Claim~3.4}} We provide here a
formal detailed proof of the correctness of the reduction; that is, we
prove \cref{cl:correctness}.
\label{app:xnl-reduction-proof}
\input{appendix/claim-proof}
\subsection{Example}
\label{app:xnl-reduction-example}
\input{appendix/example}

\section{Binary Encoding}
\label{sec:binary}
\input{sections/binary}

\section{Discussion and Open Problems}
\label{sec:conclusions}
\input{sections/conclusion}

\bibliography{references}

\end{document}

%% file: packages.tex
\usepackage[utf8]{inputenc}
\usepackage{hyperref}
\usepackage{enumerate}
\usepackage{mathtools}
\usepackage[dvipsnames]{xcolor}
\usepackage[colorinlistoftodos]{todonotes}
\usepackage[small]{complexity}
\newtheorem{fact}[theorem]{Fact}
\newtheoremstyle{problemstyle}{\topsep}{\topsep}{}{0pt}{\sffamily\bfseries}{}{5pt plus 1pt minus 1pt}%
  {\textcolor{lipicsGray}{$\blacktriangleright$}\nobreakspace\thmname{#1}:\thmnote{ #3}}
\theoremstyle{problemstyle}
\newtheorem*{problem}{Problem}

\usepackage{tikz}
\usetikzlibrary{snakes,positioning,arrows.meta,arrows,
decorations.pathreplacing,calligraphy,petri}

\usepackage{soul}

%% file: macros.tex
\renewcommand{\le}{\leqslant}
\renewcommand{\leq}{\leqslant}
\renewcommand{\ge}{\geqslant}
\renewcommand{\geq}{\geqslant}
\renewcommand{\setminus}{-}

\def\vec#1{\mathchoice{\mbox{\boldmath$\displaystyle#1$}}
{\mbox{\boldmath$\textstyle#1$}}
{\mbox{\boldmath$\scriptstyle#1$}}
{\mbox{\boldmath$\scriptscriptstyle#1$}}}
\newcommand{\Oh}{\mathcal{O}}
\newcommand{\onenorm}[1]{\left\lVert#1\right\rVert_1}
\newcommand{\norm}[1]{\left\lVert#1\right\rVert}
\newcommand{\abs}[1]{\lvert#1\rvert}
\newcommand{\size}[1]{\lvert#1\rvert}

\newcommand{\+}[1]{\mathbb{#1}}
\newcommand{\?}[1]{\mathcal{#1}}
\newcommand{\eqby}[1]{\mathrel{\raisebox{-.1ex}{\ensuremath{\stackrel{\raisebox{-.25ex}{\scalebox{.5}{\upshape\textrm{#1}}}}{=}}}}}
\newcommand{\eqdef}{\eqby{def}}
\newcommand{\fdim}[1]{\mathrm{fdim}\,#1}
\newcommand{\fin}{\mathrm{fin}}
\definecolor{pastelamber}{HTML}{FFECB9}
\definecolor{pastelteal}{HTML}{9FCFE1}
\definecolor{pastelrose}{HTML}{F4B3C3}

\newcommand{\set}[1]{\{#1\}}
\newcommand{\sset}{\subseteq}
\newcommand{\NN}{\mathbb{N}}
\newcommand{\ZZ}{\mathbb{Z}}
\newcommand{\closure}[2]{\ensuremath{[#1]^{\text{#2}}}}
\newcommand{\para}{\ComplexityFont{para}\textsf{-}}
\newclass{\XNL}{XNL}

\newcommand{\Vv}{\mathcal{V}}
\newcommand{\Ww}{\mathcal{W}}
\newcommand{\Zz}{\mathcal{Z}}
\newcommand{\config}[2]{#1(\vec{#2})}
\newcommand{\counter}[1]{\mathsf{#1}}
\newcommand{\run}[3]{\Run{#1}{#2}{}{#3}}
\newcommand{\Run}[4]{#1\xrightarrow{#2}_{#3}{#4}}
\newcommand{\len}[1]{|#1|}%
\newcommand{\unarysize}[1]{\mathit{size}(#1)}
\newcommand{\binarysize}[1]{\mathit{bitsize}(#1)}

\newcommand{\Dd}{\mathcal{D}}
\newcommand{\recognises}[1]{L(#1)}
\newcommand{\initialstate}{q^{\text{init}}}

\tikzstyle{temp} = [circle, fill, red, inner sep = 1pt]
\tikzstyle{state} = [circle, draw, black, line width = 0.4mm, inner sep = 1.5mm]
\tikzstyle{transition} = [-{Stealth[width=1.5mm, length=1.8mm]}, headless-transition]
\tikzstyle{small-transition} = [-{Stealth[width=1.2mm, length=1.5mm]}, headless-transition]
\tikzstyle{headless-transition} = [black, line width = 0.4mm]
\tikzstyle{sigma-part} = [draw, line width = 0.4mm, minimum width=20mm, minimum height=15mm, rounded corners = 2mm]
\tikzstyle{checking-part} = [draw, line width = 0.4mm, minimum width=20mm, minimum height=12mm, rounded corners = 2mm]
\tikzstyle{example-part} = [draw, line width = 0.4mm, minimum width=30mm, minimum height=18mm, rounded corners = 2mm]
\tikzstyle{example-checking-part} = [draw, line width = 0.4mm, minimum width=30mm, minimum height=9mm, rounded corners = 2mm]
\tikzstyle{small-state} = [circle, draw = black, line width = 0.4mm, inner sep = 0.8mm]

\newcommand{\colour}[2]{\textcolor{#1}{#2}}
\newcommand{\hstodo}[1]{\todo[inline, linecolor=orange, backgroundcolor=orange!30, bordercolor=black]{\textbf{Henry:} #1}}
\newcommand{\stodo}[1]{\todo[inline, linecolor=orange, backgroundcolor=orange!30, bordercolor=black]{\textbf{Sylvain:} #1}}

%% file: abstract.tex
We investigate the parameterised complexity of the classic
coverability problem for vector addition systems (VAS): given a finite
set of vectors $\vec V\subseteq \+Z^d$, an initial configuration $\vec
s\in \NN^d$, and a target configuration $\vec t\in \NN^d$, decide
whether starting from $\vec s$, one can iteratively add vectors from
$\vec V$ to ultimately arrive at a configuration that is larger than or
equal to $\vec t$ on every coordinate, while not observing any
negative value on any coordinate along the way.  We consider two
natural parameters for the problem: the dimension~$d$ and the size of
$\vec V$, defined as the total bitsize of its encoding.  We present
several results charting the complexity of those two
parameterisations, among which the highlight is that coverability for
VAS parameterised by the dimension and with all the numbers in the
input encoded in unary is complete for the class $\XNL$ under
PL-reductions.  We also discuss open problems in the topic, most
notably the question about fixed-parameter tractability for the
parameterisation by the size of $\vec V$.

%% file: sections/intro.tex
\subparagraph{Vector Addition Systems.}
Vector addition systems are a well-established model with a very
simple definition: a $d$-dimensional \emph{vector addition system}
(VAS)~\cite{KarpM69} is a finite set~$\vec V$ of vectors in~$\+Z^d$,
which defines a step relation between configurations in~$\+N^d$:
$\vec u\to_{\vec V}\vec u+\vec v$ for all $\vec u$ in $\+N^d$ and
$\vec v$ in $\vec V$, provided that $\vec u+\vec v$ is in~$\+N^d$. One usually interprets the $d$ coordinates of a configuration $\vec u\in \+N^d$ as $d$ {\em{counters}} that may take non-negative integer values; then $\vec V$ represents a set of allowed updates to the counter values, which can be applied only if none of the counters becomes negative.

In spite of this apparent simplicity, vector addition systems can
exhibit highly complex behaviours.  Famously, their \emph{reachability
problem}
is \ComplexityFont{ACKERMANN}-complete~\cite{LerouxS19,Leroux22,CzerwinskiO22}:
given as input a $d$-dimensional VAS~$\vec V$, an initial
configuration $\vec s\in\+N^d$, and a target configuration~$\vec
t\in\+N^d$, the reachability problem asks whether $\vec t$ is
reachable from~$\vec s\in\+N^d$ in~$\vec V$, i.e., whether $\vec
s\to^\ast_{\vec V}\vec t$.  This combination of a simple definition
with rich behaviours makes vector addition systems---along with the
equivalent model of Petri nets---well-suited whenever one needs to
model systems managing multiple discrete resources, e.g., threads in
concurrent computations, molecules in chemical reactions~\cite{Aris65, Aris68}, organisms in
biological processes~\cite{BaldanCMS10}, etc., but also as a theoretical tool involved in
establishing decidability and complexity statements for a variety of
decision problems in logic, formal languages, verification,
etc.~\cite[Section~5]{siglog/Schmitz16}.

\subparagraph{The Coverability Problem.}
In this note, we are interested in a decision problem with a less
extreme complexity: given the same input, the \emph{coverability
problem} asks whether there exists a coordinate-wise larger or equal
configuration $\vec t'\sqsupseteq\vec t$ such that $\vec
s\to^\ast_{\vec V}\vec t'$.  This relaxation of the reachability
problem was first shown decidable by Karp and
Miller~\cite{KarpM69}, before being proven \EXPSPACE-hard by
Lipton~\cite{Lipton76} and to belong to \EXPSPACE\ by
Rackoff~\cite{Rackoff78}.

Like reachability, coverability in vector addition systems
inter-reduces with numerous decision problems, notably in relation to
the automated verification of safety properties in concurrent or
distributed
systems~\cite{KanovichRS09,KaiserKW10,GantyM12,DiCosmoMZZ13,Esparza14,GeeraertsHS15,AlpernasPR19,BaumanMTZ22},
as well as reasoning on data-aware logics or
systems \cite{DeckerHLT14,GrigoreT16,AbdullaAAMR18}.  These
applications have motivated a thorough investigation of the
problem~\cite{KarpM69,Lipton76,Rackoff78,BozzelliG11,LazicS21,KunnemannMSSW23,SchmitzS24}
and the development of several tools specifically targeted at solving
it~\cite{EsparzaLMMN14,BlondinFHH16,GeffroyLS18}.

\subparagraph{Towards Parameterised Complexity.}  Rosier and
Yen~\cite{RosierY86} refined Rackoff's analysis to identify the
contribution of several parameters to the final complexity of the
coverability problem.  Among those, the main parameter driving the
complexity is the
\emph{dimension}, i.e., the number of counters of the system.  Indeed,
Rackoff obtains his result by showing a bound on the length of the
shortest covering runs from the source to the target configuration,
which is in $n^{2^{\Oh(d\log d)}}$ when parameterised by the
dimension~$d$, where $n$ is the size of the input encoded in
unary. (While \EXPSPACE-completeness holds with both a unary and a
binary encoding, the complexity landscape changes as soon as one
isolates the dimension as a parameter).

This upper bound on the length of shortest covering runs in
unary-encoded VAS was recently improved to $n^{2^{\Oh(d)}}$ by K{\"{u}}nnemann, Mazowiecki, Sch{\"{u}}tze,
Sinclair{-}Banks, and Węgrzycki~\cite[Theorem~3.3]{KunnemannMSSW23},
and this matches the $n^{2^{\Omega(d)}}$ lower bound in the family of
systems constructed by Lipton~\cite{Lipton76}.  In turn, this upper
bound on the length of the shortest covering runs entails that the
coverability problem with a unary encoding can be solved either in
non-deterministic space $2^{\Oh(d)}\cdot\log n$ or in deterministic time
$n^{2^{\Oh(d)}}$~\cite[Corollary~3.4]{KunnemannMSSW23}; moreover, the
latter time bound is also achieved by the classical \emph{backward
coverability algorithm}~\cite[Corollary~4.5]{SchmitzS24}.  Finally, as
also shown by K{\"{u}}nnemann et al.\ through a parameterised
reduction from the parameterised clique
problem \lang{p\textsf-CLIQUE}, this time bound is optimal if one
assumes the Exponential Time Hypothesis: under the ETH, no algorithm
running in deterministic time $n^{o(2^d)}$ for coverability may
exist~\cite[Theorem~4.2]{KunnemannMSSW23}.

\medskip
Surprisingly, in spite of this long line of results focusing on
coverability in vector addition systems, and in particular when
isolating the dimension as a key parameter, the question of the
parameterised complexity of the problem has not been considered
(see \cref{sec:related} for literature on other parameterisations on Petri nets,
which are not directly comparable).\clearpage%

\begin{problem}[\lang{p\textsf-dim\textsf-COVERABILITY(VAS)}]\hfill
\begin{description}
\item[\hspace{1.15em}input] a $d$-dimensional VAS $\vec V$, an initial configuration
$\vec s\in\+N^d$, and a target
configuration~$\vec t\in\+N^d$
\item[\hspace{1.15em}parameter] $d$
\item[\hspace{1.15em}question] does $\vec s$ cover~$\vec t$ in~$\vec V$, i.e., does there
exist~$\vec t'\sqsupseteq\vec t$ such that $\vec s\rightarrow_{\vec{V}}^\ast\vec t'$?
\end{description}
\end{problem}

\noindent From the viewpoint of parameterised complexity, when using a
unary encoding, the $n^{2^{\Oh(d)}}$ deterministic time bounds
from~\cite{KunnemannMSSW23,SchmitzS24} immediately yield
that \lang{p\textsf-dim\textsf-COVERABILITY(VAS)} is in \XP, while the
parameterised reduction from \lang{p\textsf-CLIQUE}
in~\cite{KunnemannMSSW23} shows $\W[1]$-hardness
(see \cref{fig:fp-hierarchies} for an overview of the parameterised
complexity classes discussed in this note).

\subparagraph{Fixed Systems and Parameterisation by Size.}
A long-standing open question on decision problems for VAS, already
raised by Hack~\cite[page~172]{Hack76} and revived in recent
years~\cite{Jecker22,DraghiciHR24,Czerwinski25}, is whether any
meaningful statements can be made about the complexity for
a \emph{fixed} VAS.  In this direction, Draghici, Haase, and
Ryzhikov~\cite{DraghiciHR24} have recently shown that there exists a
fixed~$\vec V$ such that the coverability problem (and thus also the
reachability problem) for $\vec V$ with a binary encoding of the
initial/target configurations is already \PSPACE-hard.  This
complexity of coverability in a fixed VAS is also related to a
question on the length of shortest covering runs by
Czerwiński~\cite{Czerwinski25}.  A natural way to approach these
questions within the framework of parameterised complexity is to ask
about the complexity of coverability when parameterised by
the \emph{size} of the encoding of the system, which we denote by
$\norm{\vec V}$.

\begin{problem}[$\lang{p\textsf-size\textsf-COVERABILITY(VAS)}$]\hfill
\begin{description}
\item[\hspace{1.15em}input] a $d$-dimensional VAS $\vec V$, an input configuration $\vec s\in\+N^d$, and a target
configuration~$\vec t\in\+N^d$
\item[\hspace{1.15em}parameter] $\norm{\vec V}$
\item[\hspace{1.15em}question] does $\vec s$ cover~$\vec t$ in~$\vec V$, i.e., does there
exist~$\vec t'\sqsupseteq\vec t$ such that $\vec s\rightarrow_{\vec V}^\ast\vec t'$?
\end{description}
\end{problem}

\noindent This is an a priori easier problem than the one parameterised by the
dimension, as there is a straightforward parameterised reduction from
this problem to the one of coverability parameterised by dimension
(see \cref{ex:fpred}).

\subparagraph{Contributions.}
Our main contribution in this note is to refine the bounds
and pinpoint the exact parameterised complexity
of \lang{p\textsf-dim\textsf-COVERABILITY(VAS)}: the problem is
\XNL-complete under PL reductions with a unary encoding (see~\cref{sec:unary}), and
\para\PSPACE-complete under FPT reductions with a binary encoding
(see \cref{sec:binary}).  These results are mainly applications of the
rich literature dedicated to the coverability problem in vector
addition systems: the upper bounds stem from the $2^{\Oh(d)}\cdot\log n$
non-deterministic space bounds from Rackoff's and subsequent
works~\cite{Rackoff78,RosierY86,KunnemannMSSW23},
while \para\PSPACE-hardness with a binary encoding follows from
the \PSPACE-hardness of coverability in a fixed
VAS~\cite[Corollary~2]{DraghiciHR24}.  The \XNL\ lower bound is a new
proof, which relies on a result of Wehar~\cite{Wehar16} on the
intersection non-emptiness problem for finite automata.  Importantly
given the dearth of `natural' complete problems for these two
parameterised complexity classes and in particular for~\XNL, these
results significantly enrich the library of known complete problems,
as most of the inter-reducible problems we mentioned earlier also have
natural parameters that correspond to the dimension.

As a consequence of these results, we also conclude
that \lang{p\textsf-size\textsf-COVERABILITY(VAS)}
is \para\PSPACE-complete when using a binary encoding
(see \cref{sec:binary}), and in \XNL\ when using a unary encoding
(see~\cref{sec:unary}).  Our motivation here is to restate the
questions about the complexity of coverability and reachability in
fixed VAS~\cite{Jecker22,DraghiciHR24,Czerwinski25} within the
framework of parameterised complexity, which we discuss more extensively
in \cref{sec:conclusions}.

%% file: sections/preliminaries.tex
\subsection{Parameterised Problems, Classes, and Reductions}
Let $\Sigma$ be a finite alphabet.  While a decision problem is just a
language $L \sset \Sigma^*$, a \emph{parameterised problem} is a set
$P \sset \Sigma^*\times \NN$.  If the pair $(x, k) \in \Sigma^* \times
\NN$ is an instance of a parameterised problem, we refer to $x$ as the
\emph{input} to the problem and we refer to $k$ as the
\emph{parameter}. We use the shorthand $n\eqdef|x|$ for the length of
$x$, whenever this does not create confusion.

The framework of parameterised complexity has led to rich hierarchies
of complexity classes, notably refining the distinction between \P\
and \NP\ in the presence of parameters through the class \FPT\ of
fixed-parameter tractable problems, along with classes of intractable
parameterised problems: the $\W$- and $\A$-hierarchies, and the
complexity classes~$\para\NP$ and~$\XP$; see \cref{fig:fp-hierarchies}
for a depiction of these classes.  A broad introduction to
parameterised complexity theory can be found in the book of Flum and
Grohe~\cite{flum}, see also the survey of de Haan and
Szeider~\cite{HaanS19} for a comprehensive overview of parameterised
classes of high complexity.  We provide below a basic recollection of
material that will be relevant to us, based on these two sources.

\subparagraph{Parameterised Reductions.}
An \emph{FPT reduction} (respectively a \emph{PL reduction}, which
stands for \emph{parameterised logspace)} from the parameterised
problem $P_1 \sset \Sigma_1^* \times \NN$ to the parameterised problem
$P_2 \sset \Sigma_2^* \times \NN$ is a
mapping \hbox{$R\colon \Sigma_1^* \times \NN \to \Sigma_2^* \times \NN$}
such that
\begin{enumerate}[(1)]
\item for all $(x,k) \in \Sigma_1^* \times \NN$, we have $(x, k) \in P_1$ if
  and only if $R(x,k) \in P_2$;
\item there is a computable function $f\colon \NN \to \NN$ and a
  constant $c \in \NN$ such that $R$ is computable in time $f(k)\cdot
  n^c$ (respectively computable in space $f(k) + c\cdot\log(n)$); and
\item there is a computable function $g\colon \NN \to \NN$ such that
  for all $(x, k) \in \Sigma_1^* \times \NN$, say with $R(x, k) = (x',
  k')$, we have $k' \leq g(k)$.
\end{enumerate}

\begin{figure}[tb]
  \centering
  \input{figures/parameterised-complexity-classes}
  \caption{\label{fig:fp-hierarchies}The parameterised complexity
    classes considered in this note and their known relations to the
    $\W$- and $\A$-hierarchies.  Arrows
    \tikz[baseline=-0.79ex,every node/.style={inner sep=.5pt}]{\node(C){\class{C}};\node[right=.3cm of C](D){\class{D}};\draw[->]
      (C) -> (D);} denote inclusions
    $\class{C}\subseteq\class{D}$; an arrow
    \tikz[baseline=-0.79ex,every node/.style={inner sep=.5pt}]{\node(C){\class{C}};\node[right=.7cm of
        C](D){\class{D}};\path[->] (C) edge[sloped] node[font=\tiny,fill=white]{$\func{fpt}$} (D);}  denotes an
   inclusion $\class{C}\subseteq\closure{\class{D}}{\func{fpt}}$ into
   the closure of~\class{D} under FPT reductions.}
\end{figure}

\subparagraph{Parameterised Space Complexity Classes.}
Although these classes are perhaps less well-known, the parameterised
paradigm can be also applied to space complexity.

Regarding nondeterministic logarithmic space complexity, a
parameterised problem $P \sset \Sigma \times \NN$ is in
(uniform) \class{XNL}~\cite[Proposition~18]{ChenFG03} if there is a
computable function $f\colon\NN \to \NN$ and a non-deterministic
algorithm that, given a pair $(x, k) \in \Sigma^* \times \NN$, decides
whether $(x, k) \in P$ in space at most $f(k)\cdot\log(n)$.  In
particular, $\AW[\SAT]$ is included in the closure of \XNL\ under FPT
reductions by~\cite[Proposition~23]{ChenFG03}.
When working with parameterised complexity classes
like \FPT, \XP, \para\PSPACE, etc., FPT reductions suffice, but
for \XNL, one needs to work with PL reductions.  When stating
completeness results, we always specify what type of reductions we
have in mind.

Regarding polynomial space complexity, a parameterised problem $P
\sset \Sigma^* \times \NN$ is in \para\PSPACE\ if there is a
computable function $f\colon\NN \to \Sigma^*$ and a problem $L \sset
\Sigma^*$ such that $L \in \PSPACE$ and for all instances $(x, k) \in
\Sigma^* \times \NN$, we have $(x, k) \in P$ if and only if $(x, f(k))
\in L$.  Thus, \para\PSPACE\ consists of all problems that can be
decided using at most polynomial space after some precomputation that
only involves the parameter.  One can also view \para\PSPACE\ as the
space complexity-concerned analogue of \FPT.  In fact, it is easy to
prove (see \cite[Exercise~8.40]{flum}) that a parameterised problem
$P$ is in \para\PSPACE\ if and only if it can be solved in FPT space; that is, in space bounded by $f(k)\cdot n^c$ for some
computable function $f\colon\NN\to \NN$ and $c\in \NN$.  In order to establish
\para\PSPACE-hardness under FPT reductions, it suffices to prove that
there exists some fixed value of the parameter~$k$ for which the
problem is \PSPACE-hard under polynomial-time reductions (see
\cite[Corollary~2.16]{flum}).

\subsection{Vector Addition Systems and Related Models}

As explained in the introduction, vector addition systems are
equivalent to a number of models; in particular, \emph{vector addition
systems with states} are more convenient in order to prove complexity
lower bounds, while \emph{Petri nets} are better suited for modelling
concurrent systems.  This section succinctly introduces vector
addition systems with states and Petri nets and here we argue that the
usual reduction between their versions of the coverability problem
holds also in the parameterised setting.

\subparagraph{Basic Notation.}  We use bold font for vectors.  We
index the $i$-th component of a vector $\vec{v}$ with square brackets
by writing $\vec{v}[i]$.  Given two vectors $\vec{u}, \vec{v} \in
\ZZ^d$, we write $\vec{u} \sqsubseteq \vec{v}$ if $\vec{u}[i] \leq
\vec{v}[i]$ for all $i \in \set{1,\ldots, d}$.  Given a vector
$\vec{v}\in \ZZ^d$, we define $\onenorm{\vec{v}} \eqdef
\abs{\vec{v}[1]} + \ldots + \abs{\vec{v}[d]}$ and
$\norm{\vec{v}}_\infty \eqdef \max\set{\abs{\vec{v}[1]}, \ldots,
  \abs{\vec{v}[d]}}$.  When working with unary encodings, we define
the size of a vector $\vec v\in\+Z^d$ as $\unarysize{\vec
  v}\eqdef\|\vec v\|_1$, and when using a binary encoding, as
$\binarysize{\vec v}\eqdef d\cdot(\log_2(\|\vec v\|_\infty+1))$.  We
write $\|\vec v\|$ when the encoding is implicit.  Then the
\emph{size} of a vector addition system~$\vec V$ in either encoding is
$\|\vec V\|\eqdef\sum_{\vec v\in\vec V}\|\vec v\|$, thus defining
$\unarysize{\vec V}$ when encoded in unary and $\binarysize{\vec V}$
when encoded in binary. Note that in both cases, we have $\|\vec V\|\geq d$ in
a non-trivial VAS~$\vec V$ of dimension~$d$.

\begin{remark}\label{ex:fpred}
  For an input VAS~$\vec V$ of some dimension~$d$, an initial
  configuration~$\vec s$, and a target configuration~$\vec t$, the map
  from $(\langle\vec V,\vec s,\vec t\rangle,\norm{\vec V})$ to
  $(\langle\vec V,\vec s,\vec t\rangle,d)$ is a PL reduction (and thus also
  an FPT reduction) from the
  \lang{p\textsf-size\textsf-COVERABILITY(VAS)} problem to the
  \lang{p\textsf-dim\textsf-COVERABILITY(VAS)} problem in both a
  unary and a binary encoding.
\end{remark}

\subparagraph{Vector Addition Systems with States.}
A $d$-dimensional \emph{vector addition system with states}
($d$-VASS)~\cite{HopcroftP79} is a pair $\?V= (Q, T)$ consisting of a
finite set of states $Q$ and a finite set of transitions $T \sset Q
\times \ZZ^d \times Q$.  To discuss the complexity of decision
problems for VASS, we define the size of a VASS $\Vv = (Q, T)$
as $\|\Vv\|\eqdef \size{Q} + \sum_{(p,\vec{x},q) \in
  T}\|\vec{x}\|$. %

Each index $i\in\{1,\dots,d\}$ can be seen as a counter with a
valuation in~$\+N$: a \emph{configuration} of a $d$-VASS is a pair
$(q, \vec{v}) \in Q \times \NN^d$ consisting of the current state $q$
and the current counter values $\vec{v}$; we use the notation
$\config{q}{v}$ for configurations.  Given two configurations
$\config{p}{u}$ and $\config{q}{v}$, there is a step
$\run{\config{p}{u}}{}{\config{q}{v}}$ if there exists a transition
$t = (p, \vec{x}, q) \in T$ such that $\vec{u} + \vec{x} = \vec{v}$;
we may refer to $\vec{x}$ as the \emph{update} of the transition $t =
(p, \vec{x}, q)$, and will also write
$\run{\config{p}{u}}{t}{\config{q}{v}}$ to emphasise that
transition~$t$ was taken from $\config{p}{u}$ to $\config{q}{v}$.

A \emph{path} in a VASS is a sequence of transitions $((p_1,
\vec{x_1}, q_1), \ldots, (p_m, \vec{x_m}, q_m))$ such that $p_{i+1} =
q_i$ for every $i \in \set{1, \ldots, m-1}$; then its \emph{length}
$\len{\pi}\eqdef m$ is its number of transitions.  A \emph{run} in a
VASS is a sequence of configurations $(\config{q_0}{v_0}, \ldots,
\config{q_m}{v_m})$ such that, for every $i \in \set{1, \ldots, m}$,
$\run{\config{q_{i-1}}{v_{i-1}}}{}{\config{q_i}{v_i}}$; in other
words, $(q_{i-1}, \vec{v_i}-\vec{v_{i-1}}, q_i) \in T$.  Let $\pi =
(t_1, \ldots, t_m)$ be a path and let $(\config{q_0}{v_0}, \ldots,
\config{q_m}{v_m})$ be a run such that, for every $i \in \set{1,
  \ldots, m}$,
$\run{\config{q_{i-1}}{v_{i-1}}}{t_i}{\config{q_i}{v_i}}$. Then we
write $\run{\config{q_0}{v_0}}{\pi}{\config{q_m}{v_m}}$ to denote that
the run is performed along the path~$\pi$, and
$\run{\config{q_0}{v_0}}{\ast}{\config{q_m}{v_m}}$ if the precise path
is not relevant.

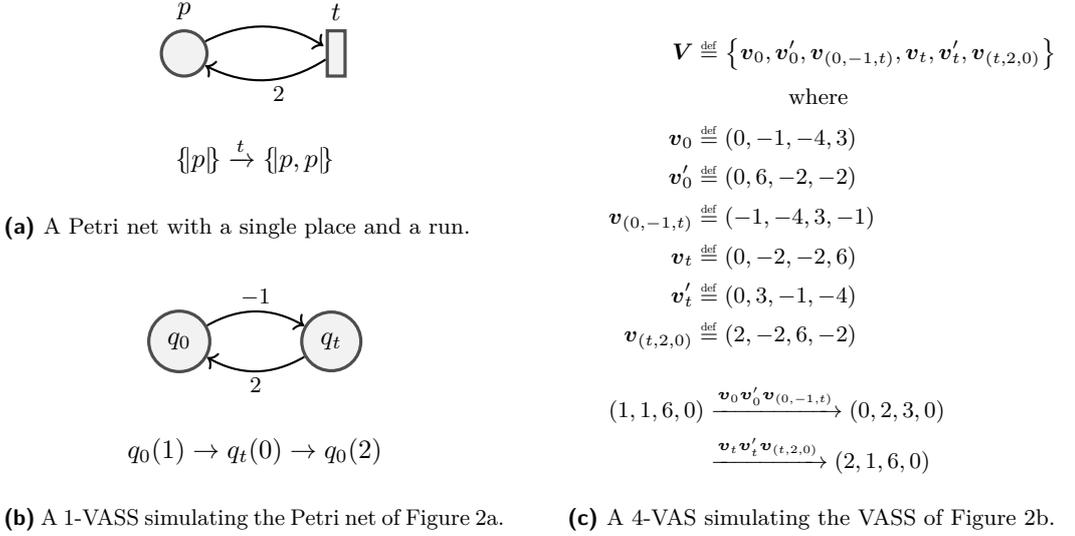
\begin{figure}[tbp]
  \begin{minipage}[t]{.47\textwidth}
  \begin{subfigure}[t]{\textwidth}
    \centering
    \begin{tikzpicture}[auto,on grid,node distance=2cm,every
        place/.style={draw=black!70,fill=black!5,very thick,minimum size=6mm}]
      \node[place,label=above:$p$](p){};
      \node[transition,right=of p,draw=black!70,fill=black!5,very thick,label=above:$t$,minimum
        width=2mm,minimum height=6mm](t){}
      edge[pre,bend right,thick](p);
      \path[->,thick] (t) edge[bend left]
      node[font=\footnotesize,inner sep=2pt]{$2$} (p);
    \end{tikzpicture}
    $$\{\!|p|\!\}\xrightarrow{t}\{\!|p,p|\!\}$$
    \caption{\label{fig:pn}A Petri net with a single place and a run.}
  \end{subfigure}\\[1.5em]
  \begin{subfigure}[t]{\textwidth}
    \centering
    \begin{tikzpicture}[auto,on grid,node distance=2cm]
      \node[state,draw=black!70,fill=black!5,very thick,minimum size=6mm](1){$q_0$};
      \node[state,draw=black!70,fill=black!5,very thick,minimum size=6mm,right=of 1](2){$q_t$};
      \path[->,thick] (1) edge[bend left] node[font=\footnotesize,inner sep=2pt]{$-1$} (2)
                (2) edge[bend left] node[font=\footnotesize,inner sep=2pt]{$2$} (1);
    \end{tikzpicture}
    $$q_0(1)\to q_t(0)\to q_0(2)$$
    \caption{\label{fig:vass}A 1-VASS simulating the Petri net of \cref{fig:pn}.}
  \end{subfigure}
  \end{minipage}\hfill
  \begin{subfigure}[t]{0.47\textwidth}
    \vspace*{-1.16cm}
    \centering\footnotesize
    \begin{align*}
      \vec V&\eqdef\big\{\vec v_0,\vec v_0',\vec
      v_{(0,-1,t)},%
      \vec v_{t},\vec v'_t,\vec v_{(t,2,0)}\big\}
      \intertext{where}
      \vec v_0&\eqdef (0,-1,-4,3)\\
      \vec v_0'&\eqdef (0,6,-2,-2)\\
      \vec v_{(0,-1,t)}&\eqdef(-1,-4,3,-1)\\
      \vec v_{t}&\eqdef(0,-2,-2,6)\\
      \vec v_t'&\eqdef(0,3,-1,-4)\\
      \vec v_{(t,2,0)}&\eqdef (2,-2,6,-2)
    \end{align*}
    \begin{align*}
      (1,1,6,0)&\xrightarrow{\vec v_0\vec v_0'\vec v_{(0,-1,t)}}(0,2,3,0)\\
      &\xrightarrow{\vec v_t\vec v_t'\vec v_{(t,2,0)}}(2,1,6,0)
    \end{align*}
    \caption{\label{fig:vas}A 4-VAS simulating the VASS of \cref{fig:vass}.}
  \end{subfigure}
  \caption{\label{fig:examples}A Petri net, its representation as a
    VASS, and its representation as a VAS, along with a run in all
    three systems.}
\end{figure}

The \emph{coverability problem} for VASS takes as input a VASS~$\Vv$,
an initial configuration $\config{p}{u}$, and a target
configuration~$\config{q}{v}$, and asks whether there exists~$\vec
v'\sqsupseteq\vec v$ such that there is a run
$\run{\config{p}{u}}{\ast}{\config{q}{v'}}$.  We denote the
parameterisations of this decision problem by
\lang{p\textsf-dim\textsf-COVERABILITY(VASS)} and
\lang{p\textsf-size\textsf-COVERABILITY(VASS)}.

A VAS can therefore be seen as a VASS with a single state, which is
dropped entirely from the definition.  Conversely, by a well-known
result from the seventies by Hopcroft and Pansiot, one can simulate
the states of a VASS at the cost of three extra dimensions in a
VAS~\cite[Lemma~2.1]{HopcroftP79}.  For clarity, the obtained VAS has
an equivalent reachability relation between configurations; a
configuration $\config{q}{x}$ in the original VASS corresponds to
configurations of the form $(\vec x, a_q, b_q, c_q)$ in the VAS, where
$a_q$, $b_q$, and $c_q$ are values bounded by $(|Q|+1)^2$ that
represent the state $q$, and such that
$(a_p,b_p,c_p)\not\sqsubseteq(a_q,b_q,c_q)$ for $p\neq q$ throughout
the computation.  Each transition of the VASS is simulated by a
sequence of three transitions of the VAS, the first two checking the
current state and the last one performing the update.  See
\cref{fig:examples} for an illustration of this construction.
\begin{fact}[c.f.\ {\cite[Lemma~2.1]{HopcroftP79}}]\label{fct:vass-to-vas}
  The parameterised problems
  \lang{p\textsf-dim\textsf-COVERABILITY(VAS)} and
  \lang{p\textsf-dim\textsf-COVERABILITY(VASS)} are equivalent up to
  PL reductions, and so are the parameterised problems
  \lang{p\textsf-size\textsf-COVERABILITY(VAS)} and
  \lang{p\textsf-size\textsf-COVERABILITY(VASS)}, all this both with a unary
  and with a binary encoding.
\end{fact}

%% file: figures/parameterised-complexity-classes.tex
  \scalebox{.68}{\begin{tikzpicture}[node distance=.7cm, semithick]
    \node[label={[color=lipicsGray,font=\footnotesize]-1:$=\class{DTIME}\big(f(k)\cdot\poly(n)\big)$}] (FPT) {\FPT};
    \node[above=of FPT] (A1) {$\W[1]=\A[1]$};
    \node[above left=.3cm and .3cm of A1] (W2) {$\W[2]$};
    \node[above right=1cm and .3cm of A1] (A2) {$\A[2]$};
    \node[above=.6cm of W2] (W3) {$\W[3]$};
    \node[above=.6cm of A2] (A3) {$\A[3]$};
    \node[above=1.8cm of W3] (WS) {$\W[\SAT]$};
    \node[above=.6cm of A3] (AWs) {$\AW[\ast]$};
    \node[above=.6cm of WS] (WP) {$\W[P]$};    
    \node[above=.6cm of AWs] (AWS) {$\AW[\SAT]$};
    \node[above=.6cm of AWS] (AWP) {$\AW[P]$};
    \node[above left=.14cm and 2.35cm of WP,
      label={[color=lipicsGray,font=\footnotesize]181:$\class{NTIME}\big(f(k)\cdot\poly(n)\big)=$}] (pNP) {$\para\NP=\para\Sigma_1^{\P}$};
    \node[above right=1.3cm and 1.8cm of AWS,
      label={[color=lipicsGray,font=\footnotesize]-1:$=\class{NSPACE}\big(f(k)\cdot\func{log}(n)\big)$}] (XNL) {\XNL};
    \node[above left=.8cm and 0.3cm of AWP] (inter) {$\XP\cap\para\PSPACE$};
    \node[above left=1.6cm and .7cm of XNL,
      label={[color=lipicsGray,font=\footnotesize]-2:$=\class{DTIME}\big(n^{f(k)}\big)$}] (XP) {\XP};
    \node[above right=2cm and .3cm of pNP,
      label={[color=lipicsGray,font=\footnotesize]181:$\class{SPACE}\big(f(k)\cdot\poly(n)\big)=$}] (pPSPACE) {$\para\PSPACE$};
    \draw[->] (FPT) -> (A1);
    \draw[->] (A1) -> (W2);
    \draw[->] (A1) -> (A2);
    \draw[->] (W2) -> (W3);
    \draw[->] (W2) -> (A2);
    \draw[->] (A2) -> (A3);
    \draw[->] (W3) -> (A3);
    \draw[dotted,->] (W3) -> (WS);
    \draw[dotted,->] (A3) -> (AWs);
    \draw[->] (AWs) -> (AWS);
    \draw[->] (WS) -> (WP);
    \draw[->] (WS) -> (AWS);
    \draw[->] (AWS) -> (AWP);
    \draw[->] (WP) -> (AWP);
    \path[->] (AWS) edge[sloped] node[near end,font=\footnotesize,fill=white,inner sep=.8pt]{$\func{fpt}$} (XNL);
    \draw[->] (XNL) -> (inter.south east);
    \draw[->] (WP) -> (pNP);
    \draw[->] (AWP) -> (inter);
    \draw[->] (inter) -> (pPSPACE);
    \draw[->] (inter) -> (XP);
    \draw[->] (pNP) -> (pPSPACE);
  \end{tikzpicture}}

%% file: appendix/petri-nets.tex
A \emph{Petri net}~\cite{Petri62} is a
tuple $\?N=(P,T,W)$ where $P$ is a finite set of places, $T$ is a
finite set of transitions, and $W{:}\,(P\times T)\cup(T\times
P)\to\+N$ is a (weighted) flow function.  It defines a transition
system with configurations in $\+N^P$, i.e.\ multisets of places (aka
\emph{markings}) that can equivalently be seen as vectors
in~$\+N^{|P|}$, and steps $\vec m\xrightarrow{t}_{\?N}\vec m'$
whenever $\vec m(p)\geq W(p,t)$ and $\vec m'(p)=\vec
m(p)-W(p,t)+W(t,p)$ for all $p$ in $P$.  We write $\vec
m\to_{\?N}^\ast\vec m'$ if the marking~$\vec m'$ is reachable from the
marking~$\vec m$ by a finite sequence of such steps.  The
\emph{dimension} of a Petri net is $|P|$, its number of places; its
\emph{size} is $\|\?N\|\eqdef |P|+|T|+\sum_{p\in P,t\in
  T}(\norm{W(p,t)}+\norm{W(t,p)})$.  The \emph{coverability problem}
for Petri nets takes as input a Petri net~$\?N$, an initial marking
$\vec m\in\+N^P$, and a target marking~$\vec m'$, and asks
whether there exists~$\vec m''\sqsupseteq\vec m'$ such that $\vec
m\to^\ast_{\?N}\vec m''$.  We denote the parameterisations of this
decision problem by \lang{p\textsf-dim\textsf-COVERABILITY(PN)} and
\lang{p\textsf-size\textsf-COVERABILITY(PN)}.

See~\cref{fig:pn} for a depiction of a Petri net $(\{p\},\{t\},W)$
with $W(p,t)=1$ and $W(t,p)=2$ as a directed bipartite graph, with
places represented by circles, transitions by rectangles, and arcs
from places to transitions representing the input flow~$W(p,t)$ and
arcs from transitions to places representing the output
flows~$W(t,p)$; the absence of an arc denotes a flow of~$0$, and an
arc without weight denotes a flow of~$1$.

As is well-known, a Petri net~$(P,T,W)$ can be encoded as an
equivalent $|P|$-dimensional VASS with $|T|+1$ states (see~\cref{fig:vass}), and conversely a $d$-dimensional VAS~$\vec V$ can be
encoded as an equivalent Petri net with $d$ places and one transition
per vector $\vec v\in\vec V$, with input flow the absolute values of
the negative values in~$\vec v$ and output flow the positive values
in~$\vec v$.  Thus we have the following equivalence in terms of
parameterised problems.
\begin{fact}\label{fct:pn-to-vas}
  The problems
  \lang{p\textsf-dim\textsf-COVERABILITY(VAS)} and
  \lang{p\textsf-dim\textsf-COVERABILITY(PN)} are equivalent up to
  PL reductions, and so are the problems
  \lang{p\textsf-size\textsf-COVERABILITY(VAS)} and
  \lang{p\textsf-size\textsf-COVERABILITY(PN)}, all this both with a unary
  and with a binary encoding.
\end{fact}

%% file: sections/known-complexity.tex
\subsection{The Complexity of Coverability}\label{sec:complexity}

\subparagraph{Upper Bounds.}
Regarding the complexity upper bounds, consider an input of the
coverability problem for a $d$-dimensional VAS~$\vec V$ with initial
configuration~$\vec s$ and target configuration~$\vec t$, encoded in
unary with size $n=\unarysize{\vec V}+\onenorm{\vec s}+\onenorm{\vec
  t}$.  The approach pioneered by Rackoff~\cite{Rackoff78} to
establish the \EXPSPACE\ upper bound is the following, and will also
be pertinent for the discussion in \cref{sec:conclusions}.  A finite
sequence of steps $\vec s\to^\ast_{\vec V}\vec t'$ for some $\vec
t'\sqsupseteq\vec t$ is a \emph{covering run} or \emph{coverability
witness}; its length is its number of steps.  What Rackoff showed is
that, if there is a covering run, then there exists one of length
bounded by~$n^{2^{\Oh(d\log d)}}$.  The best known (and actually optimal)
upper bound for this length is the following.
\begin{fact}[{\cite[Theorem 3.3]{KunnemannMSSW23}}]
    \label{thm:rackoff}
    Shortest covering witnesses for an instance of coverability
    encoded in unary have length bounded by~$n^{2^{\Oh(d)}}$.
\end{fact}

As a direct corollary, using algorithms that exhaustively explore the
space of bounded configurations, one can determine the complexity of
coverability in VAS.
\begin{corollary}[cf.~{\cite[Corollary 3.4]{KunnemannMSSW23}}]
    \label{cor:rackoff-algorithms}
  There exists both a deterministic $n^{2^{\Oh(d)}}$-time algorithm
  and a non-deterministic $2^{\Oh(d)} \cdot \log(n)$-space algorithm
  for coverability, where $n$ is the size of the input encoded in
  unary.
\end{corollary}
In the multi-parameter analysis of the complexity, one usually focuses
on the deterministic algorithm, which yields that coverability can be
solved in pseudo-polynomial deterministic time when the dimension is
fixed. In the parameterised setting, this also yields that
\lang{p\textsf-dim\textsf-COVERABILITY(VAS)} with a unary encoding is in~\XP.

In this paper however, we are more interested in the other, non-deterministic algorithm,
which we repeat from~\cite[Corollary~3.4]{KunnemannMSSW23} for the
sake of completeness. The algorithm starts with the initial configuration and
iteratively guesses the consecutive transitions in a run while keeping track
of the current configuration and the total number of transitions
applied so far (all the numbers are encoded in binary). If during
this iteration, it encounters a configuration covering the target
configuration, then it accepts, and if the bound on the length of
$n^{2^{\Oh(d)}}$ provided by \cref{thm:rackoff} is exceeded without
encountering a covering configuration, it rejects.  The correctness of
the algorithm follows directly from \cref{thm:rackoff}.  As for the
space complexity, observe that every configuration encountered during
the run will have all the counters bounded by $n^{2^{\Oh(d)}}$, and
the same can be also said about the length counter.  Hence, the space
needed to store all the relevant information is bounded by $\Oh(d\cdot
\log n^{2^{\Oh(d)}})\leq 2^{\Oh(d)}\cdot \log n$.

\subparagraph{Lower Bounds.} Regarding the lower bounds,
Lipton~\cite{Lipton76} was the first to show the \EXPSPACE-hardness of
the reachability problem, with a proof that also applies to
coverability.  The idea of the construction is to define by induction
over~$d$ a family of VAS $\vec V_{\!d}$ of dimension~$\Oh(d)$ that builds
up counter values bounded by $2^{2^d}$, allowing to simulate zero
tests on those counters. This leads to a simulation of a $3$-counter machine
with counter values bounded by~$2^{2^d}$, thereby
proving \EXPSPACE-hardness.
By a slight modification of the construction, one also obtains the
following counterpart to \cref{thm:rackoff}.
\begin{fact}[c.f.\ {\cite[Theorem~3]{BozzelliG11}}]
    \label{fct:lipton}
    Shortest covering witnesses for an instance of coverability
    encoded in unary may have length as large as~$n^{2^{\Omega(d)}}$.
\end{fact}
Lipton's construction was exploited as part of a parameterised
reduction from \lang{p\textsf-CLIQUE} to
\lang{p\textsf-dim\textsf-COVERABILITY(VAS)}~\cite[Theorem~4.2]{KunnemannMSSW23},
showing the $\W[1]$-hardness of the~problem.

%% file: appendix/related.tex
Watel, Weisser, and Barth~\cite{WatelWB17} investigate an optimisation
variant of the coverability problem in weighted Petri nets: each
transition has a non-negative weight in~$\+R_{\geq 0}$, and the goal
is to find a covering witness where the sum of weights is minimal.
The parameters they consider are $\max_{p\in P,t\in T}W(p,t)$,
$\max_{p\in P,t\in T}W(t,p)$ the maximal input and output flow weights,
$\|\vec m'\|$ the size of the target marking, and the number of steps
in the sought covering witness, none of which is comparable to our
parameterisations by dimension or size of the Petri net.

Praveen~\cite{Praveen13} defines a graph associated with a Petri
net~$\?N$ whose vertex set is~$P$ --- the set of places --- and edges $\{p,p'\}$
whenever there exists a transition~$t$ with $W(p,t)>0$
and~$W(t,p')>0$.  Praveen then shows that the coverability problem
parameterised by both the size of a vertex cover and by the maximal
flow weight $\max\{W(p,t),W(t,p)\mid p\in P,t\in T\}$ is in
\para\PSPACE~\cite[Theorem~4.5]{Praveen13}.  While this is in essence
a refinement of the parameterisation by the dimension~$|P|$ of the
Petri net (which bounds the size of a vertex cover), the added flow
weight parameter makes this result not directly comparable.

%% file: sections/p-dim.tex
\label{sec:unary}

This section is dedicated to proving the following theorem pinpointing
the exact complexity of the coverability problem in vector addition
systems, when parameterised by the dimension and using a unary
encoding.  This is a significant improvement over the \XP\ upper bound
and \W[1]-hardness mentioned in~\cite{KunnemannMSSW23}.

\begin{theorem}\label{thm:unary-xnl-complete}
  \lang{p\textsf-dim\textsf-COVERABILITY(VAS)} is \XNL-complete under
  PL reductions when using a unary encoding.
\end{theorem}

Using the reduction from \cref{ex:fpred}, this shows that the same
upper bound also applies to the parameterisation by size.
\begin{corollary}\label{cor:unary-xnl}
  \lang{p\textsf-size\textsf-COVERABILITY(VAS)} is in \XNL\ when using
  a unary encoding.
\end{corollary}

The membership to \XNL\ follows from Rackoff's upper bound, that
is,~\cref{thm:rackoff}.  The \XNL\ lower bound, however, requires some
attention.  We prove it using a reduction from the intersection
non-emptiness problem for deterministic finite automata, which is
presented in~\cref{lem:xnl-reduction} below.  To state it, we first
need to establish terminology related to automata.

\subparagraph{Finite Automata.}
A \emph{deterministic finite automaton} (DFA) $\Dd = (\Sigma, Q,
T, \initialstate, F)$ over a finite alphabet $\Sigma$, consists of a
finite set~$Q$ of states, a set $T \sset Q \times \Sigma \times Q$ of
transitions, an initial state $\initialstate \in Q$, and a set
$F \sset Q$ of \emph{final states}, with the restriction that, for
every state~$p \in Q$ and letter~$\sigma \in \Sigma$, there is at most
one state~$q \in Q$ such that $(p, \sigma, q) \in T$.  %
A \emph{run}
for a word $w=\sigma_1\sigma_2\cdots\sigma_n$ is a sequence
$(q_0,\sigma_1,q_1),(q_1,\sigma_2,q_2)\cdots(q_{n-1},\sigma_n,q_n)$ of
transitions in~$T$;
it is \emph{accepting} if $q_n\in F$ is final.  The \emph{language}
$L(\?D)\subseteq\Sigma^\ast$ recognised by~$\?D$ is the set of words
for which there exists some accepting run.  We shall assume that a DFA
is specified by its adjacency matrix, and therefore that
the \emph{size} of a given DFA $\Dd = (\Sigma, Q, T, \initialstate,
F)$ is $\size{\Dd}\eqdef\size{Q} \cdot \size{\Sigma}$.

Our interest for DFAs stems from their intersection non-emptiness
problem, a well-known \PSPACE-complete problem, which when
parameterised by the number of DFAs is one the very few
known \XNL-complete problems in the literature~\cite[Corollary
3.5]{Wehar16}.
\begin{problem}[\lang{p\textsf-INTERSECTION\textsf-NONEMPTINESS(DFA)}]\hfill
\begin{description}
\item[\hspace{1.15em}input] $k$ DFAs $\?D_1$, \dots, $\?D_k$ over the
same alphabet~$\Sigma$
\item[\hspace{1.15em}parameter] $k$
\item[\hspace{1.15em}question] is $\bigcap_{1\leq i\leq
k}L(\?D_i)\neq\emptyset$, i.e., does there exist a word
$w\in\Sigma^\ast$ with an accepting run in each~$\?D_i$?
\end{description}
\end{problem}

\subparagraph{Reduction to Coverability.}  The lower bound
in \cref{thm:unary-xnl-complete} follows from the PL reduction
from \lang{p\textsf-INTERSECTION\textsf-NONEMPTINESS(DFA)}
to \lang{p\textsf-dim\textsf-COVERABILITY(VASS)} captured by the
following statement.

\begin{lemma}\label{lem:xnl-reduction}
    Let $\Dd_1, \ldots, \Dd_k$ be DFAs over a common alphabet~$\Sigma$.
    One can compute a $2k$-VASS $\Vv$, an initial configuration
    $\config{p}{u}$, and a target configuration $\config{q}{v}$ in
    space $\Oh(\log(k) + \log(|\Dd_1| + \ldots + |\Dd_k|))$ such that $\bigcap_{1\leq i\leq k}L(\?D_i)\neq\emptyset$ if and only if there exists $\vec{v'} \geq \vec{v}$ and a run from $\config{p}{u}$ to $\config{q}{v'}$ in $\Vv$.
\end{lemma}

\begin{proof}
    Suppose $\Dd_i = (\Sigma, Q_i, T_i, \initialstate_i, F_i)$ for every $1 \leq i \leq k$.
    Without loss of generality, we shall assume that $s = \size{Q_1} = \ldots = \size{Q_k}$.
    It will also be convenient for us to number the states; $Q_i = \set{q_{i,1}, \ldots, q_{i,s}}$ and we may also assume without loss of generality that $\initialstate_i = q_{i,1}$.
    
    \proofsubparagraph{Overview.}  We construct a VASS $\Vv$ with $2k$ counters that come in $k$ pairs: for each $i\in \{1,\ldots,k\}$, there are counters $\counter{x_i}$ and $\counter{y_i}$ used to store the current state of automaton~$\Dd_i$.  Precisely, that the current state of $\Dd_i$ is $q_{i,j}$, for some $1\leq j\leq s$, is reflected by counter values $\counter{x_i} = j$ and $\counter{y_i} = s-j$. The construction of $\Vv$ ensures that, except for some intermediate configurations, the following invariant is maintained: for each $i\in \{1,\ldots,k\}$, we have $\counter{x_i}+\counter{y_i}=s$ and $\counter{x_i}\geq 1$.
    
    With this way of encoding of a selection of states of
    $\Dd_1,\ldots,\Dd_k$, we emulate applying a transition
    of any $\Dd_i$ through a pair of transitions in $\Vv$ as follows. Suppose we
    would like to apply a transition of $\Dd_i$ that goes from $q_{i,a}$ to $q_{i,b}$, for some $a,b\in \{1,\ldots,s\}$: then, we may apply two transitions in $\Vv$
    \begin{itemize}
    	\item first, one that subtracts $a$ from $\counter{x_i}$ and $s-a$ from $\counter{y_i}$, and
    	\item then, one that adds $b$ to $x_i$ and $s-b$ to $\counter{y_i}$. 
    \end{itemize}
	Thanks to the invariant, the first transition can be fired only if we indeed have $\counter{x_i}=a$ and $\counter{y_i}=s-a$ (and the state encoded by those counters is $q_{i,a}$), and firing it brings both counters to $0$. Then the second transition sets $\counter{x_i}=b$ and $\counter{y_i}=s-b$. Thus, firing the two transitions in $\Vv$ both verifies that the current state is $q_{i,a}$, and updates this state to~$q_{i,b}$.
	
	With this basic gadget understood, we may explain the overall construction of $\Vv$, depicted in \cref{fig:vass-structure}. We create two special states: the starting state $p$ and the ending state $q$. The initial configuration is $\config{p}{u}$ where $\vec{u}\eqdef(1,s-1,\ldots,1,s-1)$, which corresponds to setting every automaton $\Dd_i$ to its initial state. Next, for each $\sigma\in \Sigma$ we construct a {\em{$\sigma$-section}}, which is essentially a loop that goes from $p$ back to $p$, and a {\em{checking section}}, which is essentially a path from $p$ to $q$. Thus, any run from $p$ to $q$ first loops some number of times through the $\sigma$-sections, and finally proceed through the checking section to $q$. The intention is that the effect of performing a loop through the $\sigma$-section, for some $\sigma\in \Sigma$, is an update of the values of the counters of $\Vv$ that corresponds to performing a $\sigma$-transition in each of the automata $\Dd_1,\ldots,\Dd_k$. This is done within $k$ consecutive parts of the $\sigma$-section, each responsible for updating the state of a different $\Dd_i$ using the gadgets described in the previous paragraph. Thus, choosing consecutive loops through the $\sigma$-sections corresponds to choosing  consecutive letters of a word $w$ over $\Sigma$, and simulating the runs of all the automata $\Dd_1,\ldots,\Dd_k$ on $w$. Finally, we construct the checking section using a similar method so that it can be traversed if and only if each automaton $\Dd_i$ is in a final state; this verifies that $w$ is in the intersection of the languages of $\Dd_1,\ldots,\Dd_k$. Again, the checking section consist of $k$ parts, each responsible for checking the state of a different $\Dd_i$.  See \cref{app:xnl-reduction-example} for an illustration of this construction.

    \begin{figure} 
        \centering
        \input{figures/vass-structure}
        \caption{
            Overall structure of $\Vv$ in the proof of~\cref{lem:xnl-reduction}.
        }
        \label{fig:vass-structure}
    \end{figure}

    \proofsubparagraph{Formal Construction.}  We start by constructing the two special states $p$ and $q$. 	
	Next, for every $\sigma\in \Sigma$ we construct the
    $\sigma$-section. It contains one part for each
    $i \in \set{1, \ldots, k}$ defined as follows.
    Let $m_{\sigma,i}$ be the number of $\sigma$-transitions in $\Dd_i$.
    Then we introduce $m_{\sigma,i}+2$ many states: one starting state, one ending state, and one intermediate state for each $\sigma$-transition in~$\Dd_i$.
    Furthermore, %
    for every $\sigma$-transition $t = (q_{i,a}, \sigma, q_{i,b})$ in
    $\Dd_i$, we construct two transitions in~$\Vv$: the first
    transition from the starting state to the intermediate state $q_t$
    that updates the counters by subtracting $a$ from $\counter{x_i}$
    and $s-a$ from $\counter{y_i}$ (and does not update the other
    counters); the second transition from the intermediate state $q_t$
    to the ending state that updates the counters similarly by adding
    $b$ to $\counter{x_i}$ and $s-b$ to $\counter{y_i}$ (and does not
    update the other~counters). Finally, we connect all the parts of
    the $\sigma$-section into a path as follows by adding for each
    $i\in \{1,\ldots,k-1\}$, a counter-neutral (not changing the
    values of the counters) transition from the ending state of the
    $i$-th part to the starting state of the $(i+1)$-st part.  The
    section is connected to the rest of the VASS by counter-neutral transitions from $p$ to the starting state of the first part, and from the ending state of the last part to $p$.%
    
    Next, we present the construction of the checking section, which
    amounts to constructing the $i$-th part, for each
    $i \in \set{1, \ldots, k}$, each with two states: one starting
    state and one ending state.
    Further, we construct $\size{F_i}$ many transitions: one for each final state of $\Dd_i$.
    Precisely, for each final state $q_{i,f}$ of $\Dd_i$ (for some $f \in \set{1, \ldots, s}$), we construct a transition in $\Vv$ from the starting state to the ending state that subtracts $f$ from $\counter{x_i}$ and $s-f$ from $\counter{y_i}$ (and does not update the other counters).
    It is important to note that this is only possible if $\counter{x_i} = f$ and $\counter{y_i} = s-f$ (i.e., when $\Dd_i$ is currently at $f$).
    Accordingly, the only possible way to go is through the $i$-th part
    of the checking section is if %
    $\Dd_i$ is currently in a final state. We connect again all the parts of the checking section into a path by adding counter-neutral transitions from $p$ to the starting state of the first part, from the ending state of the last part to $q$, and from the ending state of the $i$-th part to the starting state of the $(i+1)$-st part, for each $i\in \{1,\ldots,k-1\}$.
    
    Finally, we set the initial configuration of $\Vv$ to $\config{p}{u}$, where $\vec{u}$ sets $\counter{x_i}=1$ and $\counter{y_i}=s-1$ for every $i\in \{1,\ldots,k\}$, and the target configuration to $\config{q}{0}$ (so $\vec{v}=\vec{0}$). The following claim verifies the correctness of the reduction.
    
    \begin{restatable}{claim}{correctness}\label{cl:correctness}
    	The following conditions are equivalent.
    	\begin{itemize}
    		\item There exists a word over $\Sigma$ that is accepted by all the automata $\Dd_1, \ldots, \Dd_k$.
    		\item There is run of $\Vv$ from configuration $\config{p}{u}$ to some configuration with state $q$. 
    	\end{itemize}
    \end{restatable}

	\medskip

	\noindent The proof of \cref{cl:correctness} follows easily from the construction along the lines sketched at the beginning of the proof; we give a more formal explanation in \cref{app:xnl-reduction-proof}.
    \proofsubparagraph{Space Complexity.} Let us finally observe that the VASS~$\Vv$ can be constructed in space
    $\Oh\big(\log (k)+ \log (|\Dd_1| + \ldots + |\Dd_k|)\big)$, because the
    states and transitions of~$\Vv$ can be computed using three pointers
    \begin{itemize}
        \item $i \in \set{1, \ldots, k}$ for the current automaton $\Dd_i$,
        \item $\sigma \in \Sigma$ for the current letter, and
        \item $q_{i,a} \in Q_i$, which together with $\sigma$ determines the transition $(q_{i,a}, \sigma, q_{i,b})$ in the deterministic automaton $\Dd_i$.
    \end{itemize}
    This sums up to $\Oh\big(\log(k \cdot (\abs{\Sigma} \cdot \abs{Q_1}
    + \ldots + \abs{\Sigma} \cdot \abs{Q_k})\big)$ space for the
    reduction; as furthermore $\abs{\Dd_i}
    = \abs{\Sigma}\cdot\abs{Q_i}$ for each~$i$, this is indeed the same as
    $\Oh\big(\log(k) + \log(\abs{\Dd_1} + \ldots + \abs{\Dd_k})\big)$ space.
\end{proof}

We may now conclude this section by proving \cref{thm:unary-xnl-complete}.

\begin{proof}[Proof of~\cref{thm:unary-xnl-complete}]
  Membership in \XNL\ follows from the algorithm in nondeterministic
  space $2^{\Oh(d)}\cdot\log n$ from \cref{cor:rackoff-algorithms}
  (c.f.\ \cite[Corollary~3.4]{KunnemannMSSW23}).  The \XNL-hardness
  follows from \cref{fct:vass-to-vas} and the PL reduction described
  by \cref{lem:xnl-reduction} along with the \XNL-hardness of
  \lang{p\textsf-INTERSECTION\textsf-NONEMPTINESS(DFA)} proven by
  Wehar~{\cite[Corollary 3.5]{Wehar16}}.
\end{proof}

%% file: figures/vass-structure.tex
\begin{tikzpicture}
	\node[sigma-part, fill = pastelamber] (p11) at (3,0) {\shortstack{{\small\bf Part 1}\\{}\\
	{\footnotesize$\sigma_1$-transitions}\\{\footnotesize of $\Dd_1$}}};
	\node[sigma-part, fill = pastelteal] (p12) at (5.5,0) {\shortstack{{\small\bf Part 2}\\{}\\
	{\footnotesize$\sigma_1$-transitions}\\{\footnotesize of $\Dd_2$}}};
	\node[sigma-part, fill = pastelrose] (p1k) at (9,0)
             {\shortstack{{\small\bf Part {\boldmath $k$}}\\{}\\
	{\footnotesize$\sigma_1$-transitions}\\{\footnotesize of $\Dd_k$}}};
	\node (p1d) at (7.25, 0) {$\cdots$};
	\draw[transition] (p11) -- (p12);
	\draw[transition] (p12) -- (p1d);
	\draw[transition] (p1d) -- (p1k);

	\node (l11) at (10.3, -0.5) {};
	\node (l12) at (5, -1) {};
	\draw[headless-transition, rounded corners] (p1k.east) -| (l11.center);
	\draw[headless-transition, rounded corners] (l11.center) |- (l12.center);

	\node[sigma-part, fill = pastelamber] (pl1) at (3,-2.5) {\shortstack{{\small\bf Part 1}\\{}\\
	{\footnotesize$\sigma_\ell$-transitions}\\{\footnotesize of $\Dd_1$}}};
	\node[sigma-part, fill = pastelteal] (pl2) at (5.5,-2.5) {\shortstack{{\small\bf Part 2}\\{}\\
	{\footnotesize$\sigma_\ell$-transitions}\\{\footnotesize of $\Dd_2$}}};
	\node[sigma-part, fill = pastelrose] (plk) at (9,-2.5)
             {\shortstack{{\small\bf Part {\boldmath $k$}}\\{}\\
	{\footnotesize$\sigma_\ell$-transitions}\\{\footnotesize of $\Dd_k$}}};
	\node (pld) at (7.25, -2.5) {$\cdots$};
	\draw[transition] (pl1) -- (pl2);
	\draw[transition] (pl2) -- (pld);
	\draw[transition] (pld) -- (plk);

	\node (l21) at (10.3, -3) {};
	\node (l22) at (5, -3.5) {};
	\draw[headless-transition, rounded corners] (plk.east) -| (l21.center);
	\draw[headless-transition, rounded corners] (l21.center) |- (l22.center);

	\node[rotate=90] at (3, -1.35) {$\cdots$};
	\node[rotate=90] at (5.5, -1.35) {$\cdots$};
	\node[rotate=90] at (9, -1.35) {$\cdots$};

	\node[checking-part, fill = pastelamber] (pc1) at (3,-4.5) {	\shortstack{{\small\bf Part 1}\\{}\\{\footnotesize Checking $F_1$}}};
	\node[checking-part, fill = pastelteal] (pc2) at (5.5,-4.5) {	\shortstack{{\small\bf Part 2}\\{}\\{\footnotesize Checking $F_2$}}};
	\node[checking-part, fill = pastelrose] (pck) at (9,-4.5) {
          \shortstack{{\small\bf Part {\boldmath $k$}}\\{}\\{\footnotesize Checking $F_k$}}};
	\node (pcd) at (7.25, -4.5) {$\cdots$};
	\draw[transition] (pc1) -- (pc2);
	\draw[transition] (pc2) -- (pcd);
	\draw[transition] (pcd) -- (pck);

	\node[state] (p) at (0.5, -2.5) {\large$p$};
	\draw[transition, rounded corners] (p.north) |- (p11.west) {};
	\draw[transition] (p) -- (pl1.west);
	\draw[transition, rounded corners] (p.south) |- (pc1.west) {};
	\draw[transition] (l12.center) to [rounded corners = 1mm] (2,-1) -- (p); 
	\draw[transition] (l22.center) to [rounded corners = 1mm] (1.5, -3.5) -- (p); 

	\node[state] (q) at (11, -4.5) {\large$q$};
	\draw[transition] (pck.east) -- (q);

	\draw[decoration={calligraphic brace, amplitude=5pt, mirror}, decorate, line width=1.5pt] (11.7,-0.7) -- (11.7,0.7);
	\node[align = left] at (12.75, 0) {$\sigma_1$-section};
	\draw[decoration={calligraphic brace, amplitude=5pt, mirror}, decorate, line width=1.5pt] (11.7,-3.2) -- (11.7,-1.8);
	\node[align = left] at (12.75, -2.5) {$\sigma_\ell$-section};
	\draw[decoration={calligraphic brace, amplitude=5pt, mirror}, decorate, line width=1.5pt] (11.75,-5.1) -- (11.75,-3.9);
	\node[align = left] at (12.7, -4.5) {\shortstack{checking\\section}};
\end{tikzpicture}

%% file: appendix/claim-proof.tex
\begin{claimproof} We verify the two implications in order.
	
	\proofsubparagraph{First direction.} 
	We  first argue that if there is a word in accepted all the automata $\Dd_1,\ldots,\Dd_k$, then there exists a run in $\Vv$ from $\config{p}{u}$ to $\config{q}{v}$ for some $\vec{v} \geq \vec{0}$.
	In fact, we will show that there is a run from $\config{p}{u}$ to $\config{q}{0}$.

	Suppose $w \in \Sigma^*$ is a word that is accepted by $\Dd_i$ for every $i \in \set{1, \ldots, k}$.
	We shall use vector-indexing notation ($w[1], w[2], \ldots$) to refer to the letters that comprise $w$. We construct a run of~$\Vv$ starting from $\config{p}{u}$ as follows. First, for each consecutive $j \in \set{1, \ldots, \size{w}}$, the run proceed through the $w[j]$-section as follows.
	Suppose that after the first $j-1$ letters have been read, the current state of $\Dd_i$ is $q_{i,a}$ (for some $1 \leq a \leq s$).
	Since $w$ is accepted by every automaton $\Dd_i$, there is some $w[j]$-transition in $\Dd_i$ that is taken; suppose it is the transition $t_i = (q_{i,a}, w[j], q_{i,b})$.
	Then when the run traverses the $w[j]$-section from $p$ back to $p$ by using the pair of transitions that go through the intermediate state $q_{t_i}$, for each $i\in \{1,\ldots,k\}$.
	It is possible to indeed take these transitions because if $\Dd_i$ was in state $q_{i,a}$, then $\counter{x_i} = a$ and $\counter{y_i} = s-a$.
	Then, after $t_i$ is taken in $\Dd_i$, the current state becomes $q_{i,b}$ which is reflected in the run in $\Vv$ since, after the second transition in the pair is taken, $\counter{x_i} = b$ and $\counter{y_i} = s-b$ is achieved.

	Lastly, once every letter $w[1], \ldots, w[\size{w}]$ has been read, we know that the current state of $\Dd_i$ is a final state, for every $i$.
	Precisely, suppose that the current state of each automaton $\Dd_i$ is $q_{i,f_i}$, for some $1 \leq f_i \leq s$.
	Thus, in $\Vv$, the configuration $\config{p}{z}$ that is reached has the vector of counter values $\vec{z}$ for which $\counter{x_i} = f_i$ and $\counter{x_i} = s-f_i$.
	This means that in each of the parts of the checking section, a transition can be taken to subtract $f_i$ from $\counter{x_i}$ and $s-f_i$ from $\counter{y_i}$. Thus, the constructed run can traverse the checking section and at the end reach configuration $\config{q}{0}$, as required.

	\proofsubparagraph{Second direction.}
	We now prove that if there is a run in $\Vv$ from $\config{p}{u}$ to $\config{q}{v}$ for some $\vec{v} \geq \vec{0}$, then there exists a word that is accepted by each automaton $\Dd_i$, for $i \in \set{1, \ldots, k}$.

	Let $\pi$ be a path such that
        $\run{\config{p}{u}}{\pi}{\config{q}{v}}$.  We may split $\pi$
        into subpaths $\pi_1, \ldots, \pi_m,\pi_{m+1}$ such that for
        every $j \in \set{1, \ldots, m}$, $\pi_j$ is a path from $p$
        back to $p$ that only visits $p$ at the start and the end, and
        $\pi_{m+1}$ is a path from $p$ to $q$.  Due to the overall
        structure of $\Vv$ (see~\cref{fig:vass-structure}), we know
        that for all $j \in \set{1, \ldots, m}$, $\pi_j$ is a path
        inside a $\sigma_j$-section, for some $\sigma_j\in \Sigma$,
        whereas $\pi_{m+1}$ traverses the checking section. We will
        argue that the word
	\[w\eqdef \sigma_1\,\sigma_2\,\cdots \,\sigma_m\]
	is accepted by every automaton $\Dd_i$, for $i\in \{1,\ldots,k\}$.

	For every $j\in \{1,\ldots,m\}$, let
	$\config{p}{u_j}$ be the intermediate configuration between $\pi_{j}$ and~$\pi_{j+1}$. That is, we have:
	\begin{equation*}
		\config{p}{u} 
		\xrightarrow{\pi_1} \config{p}{u_1}
		\xrightarrow{\pi_2} \config{p}{u_2}
		\xrightarrow{} \cdots \xrightarrow{} \config{p}{u_{m-1}}
		\xrightarrow{\pi_{m}} \config{p}{u_{m}}
		\xrightarrow{\pi_{m+1}} \config{q}{v}.
	\end{equation*}
	Initially, in $\vec{u}$, we know that for all $i \in \set{1, \ldots, k}$, the invariant $\counter{x_i} + \counter{y_i} = s$ and $\counter{x_i}\geq 1$ holds, and thus $\counter{x_i}$ contains the index of a state of $\Dd_i$.
	By the design of $\Vv$, we deduce that this invariant holds true also for all intermediate configurations $\config{p}{u_1}, \config{p}{u_2}, \ldots, \config{p}{u_m}$.
	
	Now, consider the sub-run $\run{\config{p}{u_j}}{\pi_j}{\config{p}{u_{j+1}}}$.
	In order to leave state $p$, go through the $\sigma_j$-section, and return to $p$, it must be the case that each part is successfully passed.
	Suppose that at in $\vec{u_j}$, it is true that $\counter{x_i} = a$ (and, due to the invariant, we know that $\counter{y_i} = s-a$).
	In order to pass the $i$-th part, it must be true that there must be a transition, say $t$, that subtracts $a$ from $\counter{x_i}$ and $s-a$ from $\counter{y_i}$.
	This is true because all transitions in the $i$-th part subtract $x$ from $\counter{x_i}$ and $s-x$ from $\counter{y_i}$ for some $1 \leq x \leq s$.
	Inside the $i$-th part of the $\sigma_j$-section, after $t$ is taken, there is a following transition, say $t'$, that adds $b$ to $\counter{x_i}$ and $s-b$ to $\counter{y_i}$ for some $1\leq b\leq s$.
	Now since $t$ and $t'$ are present in $i$-th part of the $\sigma_{f(j)}$-section, then we know there is a transition $(q_{i,a}, \sigma_{f(j)}, q_{i,b})$ in $\Dd_i$.
	Indeed, the run taking $t, t'$ is simulating $\Dd_i$ using $(q_{i,a}, \sigma_{j}, q_{i,b})$ (and so $\Dd_i$ has now read the $j$-th letter of the word). A straightforward induction now shows that for each $j\in \{1,\ldots,m\}$, vector $\vec{u}_j$ encodes, in its counters, the states in which the automata $\Dd_1,\ldots,\Dd_k$ are after reading the first $j$ letters of $w$.

	We conclude by analysing the final sub-run $\run{\config{p}{u_m}}{\pi_m}{\config{q}{v}}$, and arguing that the state in which each automaton $\Dd_i$ is after reading $w$ must be a final state. 
	This holds for the same reason used before: the transition that subtracts $a$ from $\counter{x_i}$ and $s-a$ from $\counter{y_i}$ in the $i$-th part of the checking section is only present when $q_{i,a}$ is a final state of $\Dd_i$.
	Thus, in order to reach $q$, one transition from each part of the checking section must be taken, which can only be true when the current state of automaton $\Dd_i$ is a final state.
	Thus, the word $w = \sigma_1 \, \sigma_2 \, \ldots \, \sigma_m$ is accepted by all the automata $\Dd_1,\ldots,\Dd_k$.
\end{claimproof}

%% file: appendix/example.tex
Let us illustrate the construction on the three DFAs over the alphabet
$\{a,b\}$ of \cref{fig:dfas}, and construct a $6$-VASS for which
coverability between two specified configurations holds if and only if
the intersection of the languages of the three DFAs is non-empty.

The languages of the three DFAs are the following.
\begin{itemize}
	\item $\Dd_1$ accepts all words of even length. 
	\item $\Dd_2$ accepts all words that contain at least one `$a$' and at least one `$b$'.
	\item $\Dd_3$ accepts all non-empty words where the first and second letter is not `$b$'.
\end{itemize}
The intersection of these languages is not empty;
the shortest words in the intersection of their languages are `$aaab$' and `$aaba$'.
\begin{figure} 
    \centering
    \input{figures/dfas}
    \caption{
        Three DFAs $\Dd_1$, $\Dd_2$, and $\Dd_3$.
    }
    \label{fig:dfas}
\end{figure} 

The resulting $6$-VASS $\Vv$ is presented in~\cref{fig:6vass}.
We set $s = 4$ to be the greatest number of states of any of the DFAs.
The VASS has three sections: the $a$-section (the top row of the figure), the $b$-section (the middle row of the figure), and the checking section (the bottom row of the figure).
There are three parts in each section, one for each of the three DFAs.

First, let us consider the second part in the $a$-section; here there is one pair of transitions for each $a$-transition in $\Dd_2$.
For example, consider the transitions \colour{red}{\emph{coloured in red}} in~\cref{fig:dfas} and~\cref{fig:6vass}.
The red pair of transitions in $\Vv$ corresponds to the red $a$-transition from the first state (top left) to the second state (top right) in $\Dd_2$.
The two transitions in $\Vv$ have the updates $(0,0,-1,-3,0,0)$ and $(0,0,2,2,0,0)$, respectively.
The first red transition in $\Vv$ checks if the current state of $\Dd_2$ is the first state ($\counter{x_2} = 1$ and $\counter{y_2} = s-1 = 3$).
The second red transition in $\Vv$ then updates the counters so that the current state of $\Dd_2$ is the second state ($\counter{x_2} = 2$ and $\counter{y_2} = s-2 = 2$).

Second, let us consider the third part of the checking section; here there are two transitions, one for each of the two final states of $\Dd_3$.
If the current counter values have $\counter{x_3} = 2$ and $\counter{y_3} = s-2 = 2$ or the counter values have $\counter{x_3} = 4$ and $\counter{y_3} = s-4 = 0$, then it is would be possible to pass through this part.
Indeed, when the counter values of $\Vv$ contain $\counter{x_3} = 2$ and $\counter{y_3} = 2$, then the current state of $\Dd_3$ is the second state (top right) which is a final state.
Similarly, when the counter values of $\Vv$ contain $\counter{x_3} = 4$ and $\counter{y_3} = 0$, then the current state of $\Dd_3$ is the fourth state (bottom right) which is also a final state.

As the intersection of the languages recognised by the DFAs is non-empty, we know by \cref{cl:correctness} that there is a run from $p(1,3,1,3,1,3)$ to $q(0,0,0,0,0,0)$.
The run corresponding to the word $aaab$ follows a path of over 50
transitions; listing every intermediate configuration of the run is
impractical, so we list below the main features of the run.
\begin{enumerate}[1.]
	\item From $p(1,3,1,3,1,3)$, pass through the $a$-section to reach $p(2,2,2,2,2,2)$.
	\item From $p(2,2,2,2,2,2)$, pass through the $a$-section to reach $p(1,3,2,2,4,0)$.
	\item From $p(1,3,2,2,4,0)$, pass through the $a$-section to reach $p(2,2,2,2,4,0)$.
	\item From $p(2,2,2,2,4,0)$, pass through the $b$-section to reach $p(1,3,4,0,4,0)$.
	\item From $p(1,3,4,0,4,0)$, pass through the checking section to finally reach $q(0,0,0,0,0,0)$.
\end{enumerate}

\begin{figure} 
    \centering
    \input{figures/xnl-vass-example}
    \caption{
        The $6$-VASS $\Vv$ constructed from the DFAs of \cref{fig:dfas}.
    }
    \label{fig:6vass}
\end{figure}

%% file: figures/dfas.tex
\begin{tikzpicture}
	\tikzstyle{state} = [circle, draw, black, line width = 0.4mm, inner sep = 1.5mm]
	\tikzstyle{transition} = [-{Stealth[width=1.5mm, length=1.8mm]}, black, line width = 0.4mm]

	\draw[line width = 0.4mm, rounded corners, fill = pastelamber] (-1, 1.25) rectangle (2.5, -1);
	\node[scale=1.25] at (0.75, 1.75) {$\Dd_1$};
	\node[state] (a1) at (0.2,0.125) {};
	\draw[-{Stealth[width=1.5mm, length=1.8mm]}, black, line width = 0.4mm] (a1) ++(-0.8,0) -- (a1);
	\node[circle, draw, black, line width = 0.4mm, inner sep = 1mm] at (0.2,0.125) {};
	\node[state] (a2) at (1.8,0.125) {};
	\draw[-{Stealth[width=1.5mm, length=1.8mm]}, black, line width = 0.4mm] (a1) edge[bend left = 15] node[above]{\small $a,b$} (a2);
	\draw[-{Stealth[width=1.5mm, length=1.8mm]}, black, line width = 0.4mm] (a2) edge[bend left = 15] node[below] {\small $a,b$} (a1);

	\draw[line width = 0.4mm, rounded corners, fill = pastelteal] (3, 1.25) rectangle (7.5, -1);
	\node[scale=1.25] at (5.25, 1.75) {$\Dd_2$};
	\node[state] (b1) at (4.5,0.75) {};
	\draw[-{Stealth[width=1.5mm, length=1.8mm]}, black, line width = 0.4mm] (b1) ++(-0.8,0) -- (b1);
	\node[state] (b2) at (6,0.75) {};
	\node[state] (b3) at (4.5,-0.5) {};
	\node[state] (b4) at (6,-0.5) {};
	\node[circle, draw, black, line width = 0.4mm, inner sep = 1mm] at (6,-0.5) {};
	\draw[transition, red] (b1) edge node[above]{\small $a$} (b2);
	\draw[transition] (b1) edge node[left]{\small $b$} (b3);
	\draw[transition] (b2) edge[loop below, in = -25, out = 25, distance = 8mm] node[right]{\small $a$} (b2);
	\draw[transition] (b2) edge node[right]{\small $b$} (b4);
	\draw[transition] (b3) edge[loop below, in = 155, out = 205, distance = 8mm] node[left]{\small $b$} (b3);
	\draw[transition] (b3) edge node[below]{\small $a$} (b4);
	\draw[transition] (b4) edge[loop below, in = -25, out = 25, distance = 8mm] node[right]{\small $a,b$} (b4);

	\draw[line width = 0.4mm, rounded corners, fill = pastelrose] (8, 1.25) rectangle (12.5, -1);
	\node[scale=1.25] at (10.25, 1.75) {$\Dd_3$};
	\node[state] (c1) at (9.5,0.75) {};
	\draw[-{Stealth[width=1.5mm, length=1.8mm]}, black, line width = 0.4mm] (c1) ++(-0.8,0) -- (c1);
	\node[state] (c2) at (11,0.75) {};
	\node[circle, draw, black, line width = 0.4mm, inner sep = 1mm] at (11, 0.75) {};
	\node[state] (c3) at (9.5,-0.5) {};
	\node[state] (c4) at (11,-0.5) {};
	\node[circle, draw, black, line width = 0.4mm, inner sep = 1mm] at (11,-0.5) {};
	\draw[transition] (c1) edge node[above]{\small $a$} (c2);
	\draw[transition] (c1) edge node[left]{\small $b$} (c3);
	\draw[transition] (c2) edge node[below, pos = 0.4]{\small $b$} (c3);
	\draw[transition] (c2) edge node[right]{\small $a$} (c4);
	\draw[transition] (c3) edge[loop below, in = 155, out = 205, distance = 8mm] node[left]{\small $a,b$} (c3);
	\draw[transition] (c4) edge[loop below, in = -25, out = 25, distance = 8mm] node[right]{\small $a,b$} (c4);
\end{tikzpicture}

%% file: figures/xnl-vass-example.tex
\begin{tikzpicture}

	\node[example-part, fill = pastelamber] (a11) at (2.5,0) {};
	\node[small-state, fill = white] (a1S) at (1, 0) {};
	\node[small-state, fill = white] (a1E) at (4, 0) {};
	\node[small-state] (a1A) at (2.5, 0.4) {};
	\node[small-state] (a1B) at (2.5, -0.4) {};
	\draw[small-transition] (a1S) to [rounded corners = 1mm] (1.2, 0.4) -- node[scale=0.45, above, pos=0.45]{$(\text{-}1,\text{-}3,0,0,0,0)$} (a1A);
	\draw[small-transition] (a1S) to [rounded corners = 1mm] (1.2, -0.4) -- node[scale=0.45, above, pos=0.45]{$(\text{-}2,\text{-}2,0,0,0,0)$} (a1B);
	\draw[small-transition] (a1A) to [rounded corners = 1mm] node[scale=0.45, above, pos=0.5]{$(2,2,0,0,0,0)$} (3.7, 0.4) -- (a1E);
	\draw[small-transition] (a1B) to [rounded corners = 1mm] node[scale=0.45, above, pos=0.5]{$(1,3,0,0,0,0)$} (3.7, -0.4) -- (a1E);

	\node[example-part, fill = pastelteal] (a12) at (6,0) {};
	\node[small-state, fill = white] (a2S) at (4.5, 0) {};
	\node[small-state, fill = white] (a2E) at (7.5, 0) {};
	\node[small-state, red] (a2A) at (6, 0.6) {};
	\node[small-state] (a2B) at (6, 0.2) {};
	\node[small-state] (a2C) at (6, -0.2) {};
	\node[small-state] (a2D) at (6, -0.6) {};
	\draw[small-transition, red] (a2S) to [rounded corners = 1mm] (4.7, 0.6) -- node[scale=0.45, above, pos=0.45]{$(0,0,\text{-}1,\text{-}3,0,0)$} (a2A);
	\draw[small-transition] (a2S) to [rounded corners = 1mm] (4.7, 0.2) -- node[scale=0.45, above, pos=0.45]{$(0,0,\text{-}2,\text{-}2,0,0)$} (a2B);
	\draw[small-transition] (a2S) to [rounded corners = 1mm] (4.7, -0.2) -- node[scale=0.45, above, pos=0.45]{$(0,0,\text{-}3,\text{-}1,0,0)$} (a2C);
	\draw[small-transition] (a2S) to [rounded corners = 1mm] (4.7, -0.6) -- node[scale=0.45, above, pos=0.45]{$(0,0,\text{-}4,0,0,0)$} (a2D);
	\draw[small-transition, red] (a2A) to [rounded corners = 1mm] node[scale=0.45, above, pos=0.5]{$(0,0,2,2,0,0)$} (7.2, 0.6) -- (a2E);
	\draw[small-transition] (a2B) to [rounded corners = 1mm] node[scale=0.45, above, pos=0.5]{$(0,0,2,2,0,0)$} (7.2, 0.2) --  (a2E);
	\draw[small-transition] (a2C) to [rounded corners = 1mm] node[scale=0.45, above, pos=0.5]{$(0,0,4,0,0,0)$} (7.2, -0.2) --  (a2E);
	\draw[small-transition] (a2D) to [rounded corners = 1mm] node[scale=0.45, above, pos=0.5]{$(0,0,4,0,0,0)$} (7.2, -0.6) --  (a2E);

	\node[example-part, fill = pastelrose] (a12) at (9.5,0) {};
	\node[small-state, fill = white] (a3S) at (8, 0) {};
	\node[small-state, fill = white] (a3E) at (11, 0) {};
	\node[small-state] (a3A) at (9.5, 0.6) {};
	\node[small-state] (a3B) at (9.5, 0.2) {};
	\node[small-state] (a3C) at (9.5, -0.2) {};
	\node[small-state] (a3D) at (9.5, -0.6) {};
	\draw[small-transition] (a3S) to [rounded corners = 1mm] (8.2, 0.6) -- node[scale=0.45, above, pos=0.45]{$(0,0,0,0,\text{-}1, \text{-}3)$} (a3A);
	\draw[small-transition] (a3S) to [rounded corners = 1mm] (8.2, 0.2) -- node[scale=0.45, above, pos=0.45]{$(0,0,0,0,\text{-}2, \text{-}2)$} (a3B);
	\draw[small-transition] (a3S) to [rounded corners = 1mm] (8.2, -0.2) -- node[scale=0.45, above, pos=0.45]{$(0,0,0,0,\text{-}3, \text{-}1)$} (a3C);
	\draw[small-transition] (a3S) to [rounded corners = 1mm] (8.2, -0.6) -- node[scale=0.45, above, pos=0.45]{$(0,0,0,0,\text{-}4, 0)$} (a3D);
	\draw[small-transition] (a3A) to [rounded corners = 1mm] node[scale=0.45, above, pos=0.5]{$(0,0,0,0,2,2)$} (10.7, 0.6) -- (a3E);
	\draw[small-transition] (a3B) to [rounded corners = 1mm] node[scale=0.45, above, pos=0.5]{$(0,0,0,0,4,0)$} (10.7, 0.2) --  (a3E);
	\draw[small-transition] (a3C) to [rounded corners = 1mm] node[scale=0.45, above, pos=0.5]{$(0,0,0,0,3,1)$} (10.7, -0.2) --  (a3E);
	\draw[small-transition] (a3D) to [rounded corners = 1mm] node[scale=0.45, above, pos=0.5]{$(0,0,0,0,4,0)$} (10.7, -0.6) --  (a3E);

	\node[example-part, fill = pastelamber] (b11) at (2.5,-2.5) {};
	\node[small-state, fill = white] (b1S) at (1, -2.5) {};
	\node[small-state, fill = white] (b1E) at (4, -2.5) {};
	\node[small-state] (b1A) at (2.5, -2.1) {};
	\node[small-state] (b1B) at (2.5, -2.9) {};
	\draw[small-transition] (b1S) to [rounded corners = 1mm] (1.2, -2.1) -- node[scale=0.45, above, pos=0.45]{$(\text{-}1,\text{-}3,0,0,0,0)$} (b1A);
	\draw[small-transition] (b1S) to [rounded corners = 1mm] (1.2, -2.9) -- node[scale=0.45, above, pos=0.45]{$(\text{-}2,\text{-}2,0,0,0,0)$} (b1B);
	\draw[small-transition] (b1A) to [rounded corners = 1mm] node[scale=0.45, above, pos=0.5]{$(2,2,0,0,0,0)$} (3.7, -2.1) -- (b1E);
	\draw[small-transition] (b1B) to [rounded corners = 1mm] node[scale=0.45, above, pos=0.5]{$(1,3,0,0,0,0)$} (3.7, -2.9) -- (b1E);

	\node[example-part, fill = pastelteal] (b12) at (6,-2.5) {};
	\node[small-state, fill = white] (b2S) at (4.5, -2.5) {};
	\node[small-state, fill = white] (b2E) at (7.5, -2.5) {};
	\node[small-state] (b2A) at (6, -1.9) {};
	\node[small-state] (b2B) at (6, -2.3) {};
	\node[small-state] (b2C) at (6, -2.7) {};
	\node[small-state] (b2D) at (6, -3.1) {};
	\draw[small-transition] (b2S) to [rounded corners = 1mm] (4.7, -1.9) -- node[scale=0.45, above, pos=0.45]{$(0,0,\text{-}1,\text{-}3,0,0)$} (b2A);
	\draw[small-transition] (b2S) to [rounded corners = 1mm] (4.7, -2.3) -- node[scale=0.45, above, pos=0.45]{$(0,0,\text{-}2,\text{-}2,0,0)$} (b2B);
	\draw[small-transition] (b2S) to [rounded corners = 1mm] (4.7, -2.7) -- node[scale=0.45, above, pos=0.45]{$(0,0,\text{-}3,\text{-}1,0,0)$} (b2C);
	\draw[small-transition] (b2S) to [rounded corners = 1mm] (4.7, -3.1) -- node[scale=0.45, above, pos=0.45]{$(0,0,\text{-}4,0,0,0)$} (b2D);
	\draw[small-transition] (b2A) to [rounded corners = 1mm] node[scale=0.45, above, pos=0.5]{$(0,0,3,1,0,0)$} (7.2, -1.9) -- (b2E);
	\draw[small-transition] (b2B) to [rounded corners = 1mm] node[scale=0.45, above, pos=0.5]{$(0,0,4,0,0,0)$} (7.2, -2.3) --  (b2E);
	\draw[small-transition] (b2C) to [rounded corners = 1mm] node[scale=0.45, above, pos=0.5]{$(0,0,3,1,0,0)$} (7.2, -2.7) --  (b2E);
	\draw[small-transition] (b2D) to [rounded corners = 1mm] node[scale=0.45, above, pos=0.5]{$(0,0,4,0,0,0)$} (7.2, -3.1) --  (b2E);

	\node[example-part, fill = pastelrose] (b12) at (9.5,-2.5) {};
	\node[small-state, fill = white] (b3S) at (8, -2.5) {};
	\node[small-state, fill = white] (b3E) at (11, -2.5) {};
	\node[small-state] (b3A) at (9.5, -1.9) {};
	\node[small-state] (b3B) at (9.5, -2.3) {};
	\node[small-state] (b3C) at (9.5, -2.7) {};
	\node[small-state] (b3D) at (9.5, -3.1) {};
	\draw[small-transition] (b3S) to [rounded corners = 1mm] (8.2, -1.9) -- node[scale=0.45, above, pos=0.45]{$(0,0,0,0,\text{-}1, \text{-}3)$} (b3A);
	\draw[small-transition] (b3S) to [rounded corners = 1mm] (8.2, -2.3) -- node[scale=0.45, above, pos=0.45]{$(0,0,0,0,\text{-}2, \text{-}2)$} (b3B);
	\draw[small-transition] (b3S) to [rounded corners = 1mm] (8.2, -2.7) -- node[scale=0.45, above, pos=0.45]{$(0,0,0,0,\text{-}3, \text{-}1)$} (b3C);
	\draw[small-transition] (b3S) to [rounded corners = 1mm] (8.2, -3.1) -- node[scale=0.45, above, pos=0.45]{$(0,0,0,0,\text{-}4, 0)$} (b3D);
	\draw[small-transition] (b3A) to [rounded corners = 1mm] node[scale=0.45, above, pos=0.5]{$(0,0,0,0,3,1)$} (10.7, -1.9) -- (b3E);
	\draw[small-transition] (b3B) to [rounded corners = 1mm] node[scale=0.45, above, pos=0.5]{$(0,0,0,0,3,1)$} (10.7, -2.3) --  (b3E);
	\draw[small-transition] (b3C) to [rounded corners = 1mm] node[scale=0.45, above, pos=0.5]{$(0,0,0,0,3,1)$} (10.7, -2.7) --  (b3E);
	\draw[small-transition] (b3D) to [rounded corners = 1mm] node[scale=0.45, above, pos=0.5]{$(0,0,0,0,4,0)$} (10.7, -3.1) --  (b3E);

	\node[example-checking-part, fill = pastelamber] (c11) at (2.5,-4.6) {};
	\node[small-state, fill = white] (c1S) at (1, -4.6) {};
	\node[small-state, fill = white] (c1E) at (4, -4.6) {};
	\draw[transition] (c1S) -- node[scale=0.45, above]{$(-1,-3,0,0,0,0)$} (c1E);

	\node[example-checking-part, fill = pastelteal] (c21) at (6,-4.6) {};
	\node[small-state, fill = white] (c2S) at (4.5, -4.6) {};
	\node[small-state, fill = white] (c2E) at (7.5, -4.6) {};
	\draw[transition] (c2S) -- node[scale=0.45, above]{$(0,0,-4,0,0,0)$} (c2E);
	\node[example-checking-part, fill = pastelrose] (c31) at (9.5,-4.6) {};
	\node[small-state, fill = white] (c3S) at (8, -4.6) {};
	\node[small-state, fill = white] (c3E) at (11, -4.6) {};
	\draw[transition] (c3S) to [rounded corners = 1mm] (8.3, -4.45) to [rounded corners = 1mm] node[scale=0.45, above=-0.03]{$(0,0,0,0,-2,-2)$} (10.7, -4.45) -- (c3E); 
	\draw[transition] (c3S) to [rounded corners = 1mm] (8.3, -4.75) to [rounded corners = 1mm] node[scale=0.45, above=-0.03]{$(0,0,0,0,-4,0)$} (10.7, -4.75) -- (c3E);
	
	\draw[transition] (a1E) -- (a2S);
	\draw[transition] (a2E) -- (a3S);
	\draw[transition] (b1E) -- (b2S);
	\draw[transition] (b2E) -- (b3S);
	\draw[transition] (c1E) -- (c2S);
	\draw[transition] (c2E) -- (c3S);

	\node[state] (p) at (0, -2.5) {\large$p$};

	\draw[transition, rounded corners] (p.north) |- (a1S) {};
	\node (ar1) at (11.3, -0.5) {};
	\node (ar2) at (2, -1.25) {};
	\draw[headless-transition, rounded corners] (a3E) -| (ar1.center);
	\draw[headless-transition, rounded corners] (ar1.center) |- (ar2.center);
	\draw[transition] (ar2.center) to [rounded corners = 1mm] (0.8, -1.25) -- (p);

	\draw[transition] (p) -- (b1S);
	\node (br1) at (11.3, -3) {};
	\node (br2) at (2, -3.75) {};
	\draw[headless-transition, rounded corners] (b3E) -| (br1.center);
	\draw[headless-transition, rounded corners] (br1.center) |- (br2.center); 
	\draw[transition] (br2.center) to [rounded corners = 1mm] (0.8, -3.75) -- (p);

	\draw[transition, rounded corners] (p.south) |- (c1S) {};
	\node[state] (q) at (12, -4.6) {\large$q$};
	\draw[transition] (c3E) -- (q);

\end{tikzpicture}

%% file: sections/binary.tex
In this section, we show the \para\PSPACE-completeness of the
coverability problem parameterised either by dimension or size, when
the input is encoded in binary.  As in the case of a unary encoding,
the upper bound follows from the $2^{O(d)}\cdot\log n$
non-deterministic space bounds from Rackoff's and subsequent
works~\cite{Rackoff78,RosierY86,KunnemannMSSW23}.  Regarding the lower
bound, we rely on the \PSPACE-hardness of coverability in a fixed VAS
shown by Draghici, Haase, and
Ryzhikov~\cite{DraghiciHR24}.

\begin{theorem}\label{thm:binary-para-pspace-complete}
  The problems \lang{p\textsf-dim\textsf-COVERABILITY(VAS)} and
  \lang{p\textsf-size\textsf-COVERABILITY(VAS)} are
  \para\PSPACE-complete under FPT reductions when using a binary
  encoding.
\end{theorem}
\begin{proof}
  Regarding the upper bound, recall that the coverability problem can
  be solved in non-deterministic space $2^{\Oh(d)}\cdot\log n$ where
  $d$ is the dimension and $n=\unarysize{\vec V}+\norm{\vec
  s}_1+\norm{\vec t}_1$ is the size of the input in
  unary~\cite[Corollary~3.4]{KunnemannMSSW23}.  This in turn provides
  an algorithm working in non-deterministic space $2^{\Oh(d)}\cdot N$
  where $N=\binarysize{\vec V}+d\big(\log_2(\norm{\vec s}_\infty+1)+\log_2(\norm{\vec
  t}_\infty+1)\big)$ is the size of the input in binary, showing
  that \lang{p\textsf-dim\textsf-COVERABILITY(VAS)} is
  in \para\PSPACE.  By the parameterised reduction
  from \cref{ex:fpred}, this also shows
  that \lang{p\textsf-size\textsf-COVERABILITY(VAS)} is
  in \para\PSPACE.

  Regarding the lower bound, \cite[Corollary~2]{DraghiciHR24} shows
  that there exists a fixed VASS $\Vv_{0}$ (thus of fixed size, and
  actually of dimension~$6$) such that the coverability problem
  for $\Vv_{0}$ is \PSPACE-hard under the assumption that the starting
  and the target configuration are encoded in binary.
  By \cref{fct:vass-to-vas}, there is a fixed VAS $\vec V_{\!0}$ (of
  dimension~$9$) with the same property, thus showing
  that \lang{p\textsf-size\textsf-COVERABILITY(VAS)}
  and \lang{p\textsf-dim\textsf-COVERABILITY(VAS)}
  are \para\PSPACE-hard under FPT reductions.\footnote{The \para\PSPACE-hardness of \lang{p\textsf-dim\textsf-COVERABILITY(VAS)} also follows from the fact that coverability in binary-encoded $2$-VASS is \PSPACE-hard~\cite{FearnleyJ15, BlondinFGHM15}, which implies that coverability in binary-encoded $5$-VAS is \PSPACE-hard. %
  }
\end{proof}

A consequence is that, while coverability with either parameterisation
is in~\XP\ when encoded in unary, if the problem with a binary
encoding were to belong to~\XP, then $\P=\PSPACE$.  Indeed, if this
were the case, then coverability in the fixed VAS $\vec V_{\!0}$ provided
by~\cite[Corollary~2]{DraghiciHR24} would be in~\P.

%% file: sections/conclusion.tex
Our results shed some light on the broader landscape of the
parameterised complexity of decision problems related to vector
addition systems.  Let us consider the coverability problem and the
reachability problem; both have meaningful parameterisations either by
dimension or by size.  A first set of questions is simply whether the
parameterised problem is in \FPT\ or not.  A second set of questions
concerns bounding the length of witness runs, and is motivated by the
discussion in \cref{sec:complexity}, where we saw that this approach
has been instrumental in the understanding of the complexity of
coverability.  Given an initial configuration $\vec s$ and a target
configuration $\vec t$ in a VAS~$\vec V$, a \emph{shortest witness} is
a run $\vec s\to^\ast_{\vec V}\vec t'$ for some $\vec
t'\sqsupseteq\vec t$ in the case of coverability or a run $\vec
s\to^\ast_{\vec V}\vec t$ in the case of reachability, such that the
number of steps in the run is minimal. For a parameter $k$ (the
dimension or the size of~$\vec V$), we will say that shortest
witnesses have \emph{fixed-parameter length} (FPT length) if there is
a computable function~$f$ and constant $c\in \NN$ such that shortest witnesses
have length at most~$f(k)\cdot n^c$.

This setup leads to four questions, that can each be asked with a
parameterisation by dimension or by size, and using a unary or a
binary encoding.
\begin{description}
\item[\hypertarget{Q1}{(CP)}] Is the parameterised coverability problem in~\FPT?
\item[\hypertarget{Q2}{(RP)}] Is the parameterised reachability problem in~\FPT?
\item[\hypertarget{Q3}{(CL)}] Are shortest coverability witnesses of FPT length?
\item[\hypertarget{Q4}{(RL)}] Are shortest reachability witnesses of FPT length?
\end{description}

\subparagraph{Negative Cases.} Quite a few of these questions have
negative answers.  In the case of parameterisation by dimension, we have the following:
\begin{itemize}
\item A positive answer to
  \hyperlink{Q1}{\textcolor{lipicsGray}{\textbf{\textsf{(CP)}}}} would
  entail $\XNL\subseteq\FPT$ in the case of a unary encoding by
  \cref{thm:unary-xnl-complete}, which would bring in particular a
  collapse of the \W- and \A-hierarchies down to \FPT, because
  $\AW[\SAT]\subseteq\closure{\XNL}{\func{fpt}}$
  by~\cite[Proposition~23]{ChenFG03} (recall
  \cref{fig:fp-hierarchies}).  It would entail $\FPT=\para\PSPACE$ in
  the case of a binary encoding by
  \cref{thm:binary-para-pspace-complete}, which is if and only if
  $\P=\PSPACE$.
\item \hyperlink{Q3}{\textcolor{lipicsGray}{\textbf{\textsf{(CL)}}}}
  has an unconditional negative answer regardless of the choice of
  encoding by \cref{fct:lipton}: Lipton's construction~\cite{Lipton76}
  yields a family of inputs to the coverability problem in
  dimension~$\Oh(d)$ where the shortest witnesses have length
  at least $n^{2^{d}}$.%
\item \hyperlink{Q2}{\textcolor{lipicsGray}{\textbf{\textsf{(RP)}}}}
  and \hyperlink{Q4}{\textcolor{lipicsGray}{\textbf{\textsf{(RL)}}}}
  have a negative answer regardless of the choice of encoding since
  the reachability problem in dimension~$\Oh(d)$ is hard for the
  fast-growing complexity class
  $\mathbf{F}_{d}$~\cite{Leroux22,CzerwinskiO22}, and the shortest
  witnesses in their instances have length at least $F_d(n)$ where
  $F_d$ is the $d$th fast-growing function (see \cite{Schmitz16} for a
  reference on fast-growing complexity).
\end{itemize}
Turning our attention to parameterisation by size and a
binary encoding, we have:
\begin{itemize}
\item A positive answer to either
  \hyperlink{Q1}{\textcolor{lipicsGray}{\textbf{\textsf{(CP)}}}} or
  \hyperlink{Q2}{\textcolor{lipicsGray}{\textbf{\textsf{(RP)}}}} would
  again entail $\P=\PSPACE$ by
  \cref{thm:binary-para-pspace-complete}.
\item Finally,
  \hyperlink{Q3}{\textcolor{lipicsGray}{\textbf{\textsf{(CL)}}}} and
  \hyperlink{Q4}{\textcolor{lipicsGray}{\textbf{\textsf{(RL)}}}} have
  a negative answer, since with a binary encoding and a fixed VAS of constant size, one can easily ensure that shortest witnesses have length exponential in the
  size of the input.
\end{itemize}

\subparagraph{Open Questions.}  
Our questions are therefore only relevant in the case of a unary encoding and a parameterisation by size.  
In that setting, 
\hyperlink{Q3}{\textcolor{lipicsGray}{\textbf{\textsf{(CL)}}}} is essentially Question~3 by Czerwiński~\cite{Czerwinski25}\footnote{Hack
  observed~\cite[page~171]{Hack76} that shortest coverability witnesses
  were of linear length if, in addition to the VAS, the initial
  configuration was fixed (thanks to the coverability tree
  construction~\cite{KarpM69}).  On a similar note, Rackoff's
  analysis~\cite{Rackoff78} even yields a constant bound on the length
  of shortest coverability witnesses if we fix both the VAS and the
  target configuration.} and
\hyperlink{Q4}{\textcolor{lipicsGray}{\textbf{\textsf{(RL)}}}} is essentially a conjecture due Hack~{\cite[page 172]{Hack76}} that has recently been restated by Jecker~\cite{Jecker22}.  
Both of these questions were originally phrased in terms of an
  $\Oh(n)$ upper bound in a fixed VAS; here we cast them in the
  terminology of parameterised complexity.

Those two questions are connected to
\hyperlink{Q1}{\textcolor{lipicsGray}{\textbf{\textsf{(CP)}}}} and
\hyperlink{Q2}{\textcolor{lipicsGray}{\textbf{\textsf{(RP)}}}}: an FPT
length on shortest witnesses implies the existence of a
nondeterministic algorithm that guesses and checks this witness. With
a unary encoding, this entails that the problem is in~\para\NP%
, while with a binary encoding, this
entails that the problem is in~\para\PSPACE%
.  Recall that all we know at the moment are \XNL\ and
\para\PSPACE\ upper bounds for coverability with unary and binary
encoding, and no reasonable upper bounds for reachability.  Thus
\hyperlink{Q3}{\textcolor{lipicsGray}{\textbf{\textsf{(CL)}}}} and
\hyperlink{Q4}{\textcolor{lipicsGray}{\textbf{\textsf{(RL)}}}} provide
an original angle of attack on \hyperlink{Q1}{\textcolor{lipicsGray}{\textbf{\textsf{(CP)}}}} and
\hyperlink{Q2}{\textcolor{lipicsGray}{\textbf{\textsf{(RP)}}}}, respectively, and a positive answer would yield
\begin{itemize}
\item the \para\PSPACE-completeness of reachability parameterised by
  size with a binary encoding, and
\item an indication that the complexity of the problems parameterised
  by size with a unary encoding is likely lower than \XNL\ or
  \para\NP.
\end{itemize}
In summary, we believe that \hyperlink{Q1}{\textcolor{lipicsGray}{\textbf{\textsf{(CP)}}}} and
\hyperlink{Q2}{\textcolor{lipicsGray}{\textbf{\textsf{(RP)}}}} %
are excellent open questions in parameterised complexity that deserve further investigation, whereas \hyperlink{Q3}{\textcolor{lipicsGray}{\textbf{\textsf{(CL)}}}} and
\hyperlink{Q4}{\textcolor{lipicsGray}{\textbf{\textsf{(RL)}}}} are auxiliary conjectures whose resolution might be very insightful for \hyperlink{Q1}{\textcolor{lipicsGray}{\textbf{\textsf{(CP)}}}} and
\hyperlink{Q2}{\textcolor{lipicsGray}{\textbf{\textsf{(RP)}}}}.

As a further indication on the easiness of coverability and
reachability parameterised by size with a unary encoding, note that in
a \emph{fixed} VAS $\vec V_{\!0}$, the coverability problem and the
reachability problem are \emph{sparse} languages, meaning that there
exists a polynomial $p$ such that, for all~$n$, the number of
instances $x$ of the language of length $|x|=n$ is bounded by~$p(n)$.
Indeed, as $\vec V_{\!0}$ and thus its dimension~$d$ are fixed, there
are at most $\binom{n}{2d}\leq n^{2d}$ possible pairs of initial and
target configurations such that $\norm{\vec s}_1+\norm{\vec t}_1=n$.
This tells us that if there exists a fixed~$\vec V_{\!0}$ such that either
coverability or reachability for $\vec V_{\!0}$ with unary encoding
is \NP-hard, then $\P=\NP$ by Mahaney's
Theorem~\cite[Theorem~3.1]{Mahaney82}. Similarly, \NL-hardness
entails $\L=\NL$~\cite[Corollary~3]{CaiS00}.